\documentclass[11pt,titlepage]{article}
\emergencystretch=\maxdimen
\usepackage[pass]{geometry}
\usepackage{CJK,CJKspace,CJKpunct}
\usepackage{bm}
\usepackage{amsfonts}
\usepackage{amsthm}
\usepackage{amsmath}
\usepackage{mathtools}  
\usepackage{amssymb}
\usepackage{graphicx}
\usepackage[colorinlistoftodos]{todonotes}
\usepackage{verbatim}
\usepackage{authblk}
\usepackage{multirow}
\usepackage{lscape,booktabs,longtable}
\usepackage[title]{appendix}
\newcommand\pdif[2]{\frac{\partial #1}{\partial #2}}
\newcommand\plim[1]{\underset{#1}{\text{plim}}}
\DeclareMathOperator*{\argmin}{\arg\min}
\DeclareMathOperator*{\argmax}{\arg\max}
\usepackage{caption}
\usepackage{subcaption}
\usepackage{natbib}
\usepackage{hyperref}
\usepackage{algorithm}
\usepackage{algorithmic}
\newcommand{\Expect}{\mathbb{E}}
\newlength{\subfigwidth}
\newlength{\subfigcolsep}
\newcommand{\nn}{\nonumber\\}
\newtheorem{theorem}{Theorem}
\newtheorem{remark}{Remark}
\newtheorem{lemma}{Lemma}

\newtheorem{assumption}{Assumption}
\newtheorem{definition}{Definition}
\newcommand{\GMM}{\mathrm{GMM}}
\newcommand{\GMMMPEC}{\mathrm{GMM-MPEC}}

\usepackage{enumitem}

\newcommand{\NS}{\mathcal{N}_S}
\newcommand{\NmS}{\mathcal{N}_{\setminus S}}
\newcommand{\NmSj}[1]{\mathcal{N}_{\setminus S_{#1}}}
\newcommand{\statNS}{R_{\mathrm{CV}}}

\newcommand{\NSSizeinv}{|\mathcal{N}_{S}|^{-1}}
\newcommand{\NSSizehalf}{|\mathcal{N}_{S}|^{1/2}}
\newcommand{\NSSizehalfinv}{|\mathcal{N}_{S}|^{-1/2}}

\newcommand{\NmSSizeinv}{|\mathcal{N}_{\setminus S}|^{-1}}
\newcommand{\NmSSizehalf}{|\mathcal{N}_{\setminus S}|^{1/2}}
\newcommand{\NmSSizehalfinv}{|\mathcal{N}_{\setminus S}|^{-1/2}}
\newcommand{\Qempi}[1]{Q_{\mathrm{valid}}^{(#1)}}
\newcommand{\Qtruei}[1]{Q_0^{(#1)}(\theta^{(#1)}_0)}
\newcommand{\QempvalidSi}[2]{Q_{S,\mathrm{valid}}^{(#1)}(#2)}
\newcommand{\sigmaemp}{\hat{\sigma}}
\newcommand{\hypnull}{\mathrm{H}_0}
\newcommand{\hypone}{\mathrm{H}_1^{(a)}}
\newcommand{\hyptwo}{\mathrm{H}_1^{(b)}}
\newcommand{\Model}[1]{\mathcal{M}_{#1}}
\newcommand{\Wtrue}[1]{W^{(#1)}}
\newcommand{\Wemp}[1]{W_S^{(#1)}}
\newcommand{\tautrue}[1]{\tau_0^{(#1)}}
\newcommand{\tauemp}[1]{\hat{\tau}_S^{(#1)}}
\newcommand{\dimf}[1]{q_{#1}}
\newcommand{\dimtheta}[1]{p_{#1}}

\newcommand{\Shalf}{S^{1/2}}

\newcommand{\Thalf}{T^{1/2}}

\newcommand{\Astari}[1]{A_*^{(#1)}}
\newcommand{\myvec}{\mathrm{vec}}
\newcommand{\vech}{\mathrm{vech}}
\newcommand{\Deltai}[1]{\Delta^{(#1)}}
\newcommand{\Rtrue}{R_*}
\newcommand{\Rtruei}[1]{R_*^{(#1)}}
\newcommand{\Rtruesingle}{{R_*^{\mathrm{single}}}}
\newcommand{\Rtruesinglei}[1]{R_*^{\mathrm{single},(#1)}}
\newcommand{\Vtrue}{V_*}
\newcommand{\Vtruesingle}{V_*^{\mathrm{single}}}
\newcommand{\combinationSetS}{\{S \subset \{1,2,\dots,r\}: |S| = r-k\}}
\newcommand{\Var}{\mathrm{Var}}
\newcommand{\Cov}{\mathrm{Cov}}
\newcommand{\rCk}{{}_r C _k}
\newcommand{\myC}[1]{c_{#1}}

\newcommand{\mNj}[1]{\mathcal{N}_{T_{j,r}}}
\newcommand{\Gtruei}[1]{G_0^{(#1)}}
\newcommand{\Gempi}[1]{G_S^{(#1)}}
\newcommand{\thetaempS}[1]{\hat{\theta}_{#1}}
\newcommand{\thetaempiS}[2]{\hat{\theta}_{#2}^{(#1)}}
\newcommand{\thetabariS}[2]{\bar{\theta}_{#2}^{(#1)}}
\newcommand{\FSi}[2]{F_{#2}^{(#1)}}
\newcommand{\CovS}{\Sigma}
\newcommand{\CovSemp}{\hat{\Sigma}}
\newcommand{\CovSij}[2]{\Sigma^{(#1)}_{#2}}

\newcommand{\CovSempij}[2]{\hat{\Sigma}^{(#1)}_{#2}}
\newcommand{\thetai}[1]{\theta^{(#1)}}
\newcommand{\thetatrue}{\theta_0}
\newcommand{\thetatruei}[1]{\theta_0^{(#1)}}
\newcommand{\Phii}[1]{\Phi^{(#1)}}
\newcommand{\mui}[1]{\mu^{(#1)}}

\newcommand{\inprob}{\overset{p}{\to}}

\begin{CJK*}{UTF8}{zhsong}
\title{Cross Validation Based Model Selection via Generalized Method of Moments}
\author[1]{Junpei Komiyama\thanks{jkomiyama@tkl.iis.u-tokyo.ac.jp}}
\author[2]{Hajime Shimao\thanks{hshimao@santafe.edu}}
\affil[1]{University of Tokyo}
\affil[2]{Santa Fe Institute}

\end{CJK*}

\begin{document}
\maketitle

\begin{abstract}
Structural estimation is an important methodology in empirical economics, and a large class of structural models are estimated through the generalized method of moments (GMM). 
Traditionally, selection of structural models has been performed based on model fit upon estimation, which take the entire observed samples.
In this paper, we propose a model selection procedure based on cross-validation (CV), which utilizes sample-splitting technique to avoid issues such as over-fitting. While CV is widely used in machine learning communities, we are the first to prove its consistency in model selection in GMM framework. Its empirical property is compared to existing methods by simulations of IV regressions and oligopoly market model. In addition, we propose the way to apply our method to Mathematical Programming of Equilibrium Constraint (MPEC) approach. Finally, we perform our method to online-retail sales data to compare dynamic market model to static model.%
\end{abstract}


\section{INTRODUCTION}\label{1-Introduction}
Structural estimation of economic models is one of the most widely used methodologies in empirical economics nowadays in variety of fields. Structural estimation enables researchers to interpret latent variable, as well as it allows researchers to perform a counterfactual simulation. Arguably, however, one of the largest shortcoming in the structural estimation procedure lies in the selection of a proper model. That is, the specification of estimation models is usually chosen by researchers 
and rarely empirically tested. The structural estimation itself does not directly address on it from the data, because the estimation is performed by assuming the model reflects the true data generating process (\cite{angrist2010}). On a paper it is a common practice for economists to verbally argue and defend their model specification in a descriptive way. 
However, since the validity of the counterfactual simulation crucially depends on the goodness of the model, verifying and choosing a proper model empirically is of particular importance. Especially, we often simplify a model for the ease of tractability: Such simplifications is preferred to be subject to some assessment.

When a structural model is estimated in economics, researchers often use generalized method of moments (GMM) as well as maximum likelihood. 
As to selecting a true model, \cite{smith1992} and \cite{riversvuong2002} offer a model selection procedure for GMM based on the difference of empirical moments. Their core idea is a simple use of the GMM minimand as a fitness of the model with the observed data: That is, to select the model of the smallest GMM minimand when it is estimated\footnote{The theory provided in \cite{riversvuong2002} applies to broader range of model selection criteria. However, it is often implemented as GMM minimand comparison. See \cite{bonnet2010inference} or \cite{berto2007vertical} for example.}. Although such a procedure is asymptotically consistent in choosing a true or "better" model, the performance of model selection with limited sample size is still uncertain. In some applications, economists have to make an inference from a relatively small number of observations. Given a limited size of the sample, their procedures may be subject to "over-fitting": excessively complicated models can fit tighter to the observations in hand with better "goodness-of-fit" criterion, and thus is selected as a better model even if the model is not very true.

To avoid over-fitting problem, some model selection criteria such as AIC-GMM or BIC-GMM "penalize" the number of parameters in a model (\cite{andrews1999}). However, the complexity of economic models is not simply measured by the number of parameters. Structural model may include non-parametric components in specification (e.g., \cite{gautier2013nonparametric}), where we cannot apply a penalization based on number of parameters. Additionally, estimation procedure sometimes involves nonparametric approximation only for certain models. For example, estimation of dynamic demand model in \cite{gowrisankaran2012dynamics} includes a nonparametric approximation of a value function, which may make their model more flexible than static demand model. To date, it is not well understood how these factors contribute to the over-fitting issue nor how to penalize its flexibility. 


In this paper, we offer a novel approach to this problem that helps researchers to identify the best model specification from the data. Our idea is to apply the cross-validation (CV) method, which is commonly used in other areas such as machine learning, in evaluating the predictive power of the model. 
The main idea behind cross-validation is to split the data into several portions so that test of a model fit is implemented on a different data from the one used for estimating parameters. As a result, the estimated moment suffers a smaller over-fitting than in-sample model selection.

The largest advantage of sample splitting lies in its wide range of potential applications. On applying CV, one does not need to take the number of model parameters explicitly. As a result, it can select the true model among parametric, non-parametric and even semi-parametric models. 
Moreover, CV can be applied not only in selecting models, but also selecting hyper-parameters of estimation and even estimation method itself. 
For example, estimation of dynamic model often includes approximation of value function on a discrete grid space, where the coarseness of the grid space has not been paid adequate attention though it heavily influences the performance of estimation.
As to the example of estimation method, random coefficient demand system can be estimated in various specifications, such as parametric or non-parametric, through various methodologies such as nested fixed point algorithm or constrained optimization approach (MPEC, \cite{su2012constrained}) and they may yield different results especially in limited sample size.

Economists typically evaluate estimation techniques and model specification by checking how the true parameters are recovered in a Monte-Carlo simulation. However, the best specification or methodology may vary across different data or the "true" data generating process that researchers do not observe. Thus, it is preferable to make an assessment in real-world data as well, and CV offers a practical approach to that end. Taking a wide range of applications into consideration, conducting CV in selecting models deserves a significant portion of attention.

Although CV is commonly used in data science fields such as machine learning and data mining, its applicability to economic models is not obvious. In machine learning and data mining, the primal concern lies in how accurate the prediction of a regressor or classifier is. Meanwhile, in empirical economics, identifying the model reflecting the reality closer and estimating its model parameters are of primal concern, and machine learning literature does not provide a sufficient guarantee in identification of a model. This gap remains to be closed in applying data science methods in econometrics. Taking this into consideration, we propose an identifiable CV method for GMM.

We first prove the consistency of cross-validation algorithm: That is, the algorithm identifies a correctly specified model from misspecified models with the probability  approaching to $1$ as the number of data increases.
When a model is estimated through likelihood maximization, \cite{yang2007consistency} proved the consistency of the cross-validation in non-parametric regression model selection.
We prove an analogous result for GMM version of CV algorithm. 

After giving the consistency, we test the performance of our cross-validation algorithm with a limited number of samples by Monte-Carlo simulation. Firstly, we examine a simple instrumental variable regression. We observe our algorithm selects a correctly specified mode over a misspecified model with high probability even when data size is limited. Importantly, our algorithm finds the correctly specified model even when the alternative model has higher flexibility (i.e., more parameters) than the true model, suggesting that it is robust to over-fitting. Furthermore, we compare the performance of our algorithm with Rivers-Vuong type GMM minimand comparison approach and also approaches based on GMM-AIC and GMM-BIC criteria that \cite{andrews1999} suggested. The result implies that the comparison of GMM minimand suffers over-fitting, and as a result it often selects a misspecified model of higher complexities. Though GMM-AIC and GMM-BIC based approaches attempt to solve the over-fitting problem by penalizing the flexibility of model, their performance turns out to be extremely sensitive to the model specification, and as a result, they often fail to find the correctly specified model.

Secondly, we conduct another experiment in more complex nonlinear models. We use a collusive pricing model similar to the ones of \cite{bresnahan1987} and \cite{hu2014collusion}, where their objective of model selection is to detect a potential tacit collusion from the sales and price data. 
We simulate the price and quantity data from perfectly competitive setting and partially collusive setting, and test if our algorithm discovers the true conduct or not. We show that our cross-validation procedure generally perform well to identify the true pricing structure from a limited amount of data. We show how CV outperforms the simple GMM fitting comparison without data split. 

In addition, we propose a method to apply cross-validation algorithm when estimation is based on Mathematical Programming of Equilibrium Constraint (MPEC)  approach. MPEC is proposed by \cite{su2012constrained} and is one of the state-of-the-art estimation methodologies.  MPEC achieves high computational efficiency by avoiding the nested fixed point algorithm, and its convenience is earning significant attention especially in the industrial organization research community. Though application of CV to MPEC is not straightforward, we provide a modified algorithm of CV applicable to MPEC estimation.

Finally, we perform our algorithm on a cutting-edge structural model with real-world data. The model we adopt is dynamic demand and dynamic pricing model of \cite{conlon2012dynamic}. The dynamic models are considered to be the recent frontier of the industrial organization community and used in many applications (such as \cite{lee2013vertical}). However, the superiority of the dynamic models compared with static models on its explainability of the consumer behavior is not sufficiently supported. Likewise, the dynamic pricing model is a frontier research topic in the industrial organization (\cite{nair2007intertemporal},\cite{luo2015operating}), but its empirical support against static model is only descriptive. We apply our CV algorithm to the market data of an online retailer based in the UK to test dynamic models against static models. We show that the results are mixed across different products, even though they are sold by the same retailer.

The paper proceeds as follows. In Section 2, we formally introduce cross-validation in GMM and discuss its econometric property. In particular, we prove its asymptotic consistency. In Section 3, we demonstrate a Monte-Carlo experiment of model selection in IV regression. In Section 4, we perform a further experiment in an oligopolistic pricing model as a nonlinear example. Section 5 explains how we can modify the algorithm when it is applied to MPEC approach. Section 6 presents the setup and results of the real-world application of the dynamic pricing model using online-retailer data. Section 7 concludes the paper.

\section{CROSS-VALIDATION APPROACH TO GMM MODEL SELECTION}\label{2-Theory}
\subsection{Setup}
Let $\mathbf{v}=\{v_t\}$ be a random vector of observed data in $\mathbf{V} \subset R^d$. Let $\mathcal{M}_i$ for $i=1,2$ be the two candidate models to explain the observed data. Each model, if correctly specified, is characterized by a set of moment conditions $f^{(i)}:\mathbf{V}\times\Theta^{(i)}\rightarrow \mathcal{R}^{q_i}$ such that
\begin{align*}
\mathcal{M}_i \Rightarrow E[f^{(i)}(v_t,\theta_0^{(i)})]=0 \text{ for a unique } \theta_0^{(i)}\in \Theta^{(i)}
\end{align*}
where $\thetai{i}\in \Theta^{(i)}$ denotes the parameters of a model $i$ to be estimated.
Let $\dimtheta{i}$ be the dimension of $\thetai{i}$
Given the observation $\{v_t\}_{t=1,,,T}$, the parameters of each model are estimated via GMM;
\begin{align}
\hat{\theta}_T^{(i)}=\argmin_{\theta^{(i)}\in\Theta^{(i)}} Q_T^{(i)}(\theta^{(i)})
\label{eq_gmmargmin}
\end{align}
where
\begin{align*}
Q_T^{(i)}(\theta^{(i)})=\left\{ \frac{1}{T} \sum_{t=1}^T f^{(i)}(v_t,\theta^{(i)}) \right\}' W_T^{(i)} \left\{ \frac{1}{T} \sum_{t=1}^T f^{(i)}(v_t,\theta^{(i)}) \right\}. 
\end{align*}

Let $\plim{T\rightarrow\infty}W_T^{i}=W^{i}$, and the population analogue of the moment conditions be
\begin{align*}
Q_0^{(i)}(\theta^{(i)})=E[f^{(i)}(v_t,\theta^{(i)})]'W^{(i)}E[f^{(i)}(v_t,\theta^{(i)})].
\end{align*}
Assume that $\plim{T\to\infty}\theta_T^{(i)}=\theta^{(i)}_0$ exists.
The null hypothesis is that $\mathcal{M}_1$ and $\mathcal{M}_2$ are asymptotically equivalent;
\begin{align*}
\text{H}_0: Q_0^{(1)}(\theta^{(1)}_0)=Q_0^{(2)}(\theta^{(2)}_0).
\end{align*}
Two alternative hypotheses are that $\mathcal{M}_1$ is asymptotically better than $\mathcal{M}_2$ or the other way around;
\begin{align*}
\text{H}_1^{(a)}=Q_0^{(1)}(\theta^{(1)}_0)<Q_0^{(2)}(\theta^{(2)}_0),\\
\text{H}_1^{(b)}=Q_0^{(1)}(\theta^{(1)}_0)>Q_0^{(2)}(\theta^{(2)}_0).
\end{align*}

\subsection{Cross-validation}

Cross-validation is a model selection procedure in which the data is split into two subsets called training set and validation set. The set of parameters of each model is trained in the trained set, and its goodness is evaluated with the validation set. Let $r \ge 2$, $k<r$ be integers. In leave-$k$-out $r$-fold cross-validation (($k,r$)-CV), we first split $T$ datapoints into $r$ disjoint subsets. At each round of CV, We use $r-k$ of them as the training data, and the other $k$ as the validation data. Multiple number of rounds among possible splits are performed to reduce variability.
Namely, let
\begin{align*}
\mathcal{N}_{T_{j,r}} = \{\lfloor T(j-1)/r \rfloor +1, \lfloor T(j-1)/r \rfloor+2,\dots,\lfloor Tj/r \rfloor\} 
\end{align*}
be the indices of the $j$-th split.
Let $\combinationSetS$ and 
\begin{align*}
\mathcal{N}_S = \bigcup_{j \in S} \mathcal{N}_{T_{j,r}}
\end{align*}
be subset of datapoints consisted of folds in $S$. 
The moment on this datapoints is denoted as
\begin{equation*}
Q_S^{(i)}(\theta^{(i)})=\left\{ \frac{1}{|\mathcal{N}_S|} \sum_{t \in \mathcal{N}_S} f^{(i)}(v_t,\theta^{(i)}) \right\}' W_S^{(i)} \left\{ \frac{1}{|\mathcal{N}_S|} \sum_{t \in \mathcal{N}_S} f^{(i)}(v_t,\theta^{(i)}) \right\},
\end{equation*}
and the model trained to minimize the moment is denoted as
\begin{align*}
\thetaempiS{i}{S}=\argmin_{\theta^{(i)}\in\Theta^{(i)}} Q_S^{(i)}(\theta^{(i)}).
\end{align*}
Once the model is trained, it is validated by the rest of datapoints as: 
\begin{equation*}
Q_{S,\mathrm{valid}}^{(i)}(\thetaempiS{i}{S})=\left\{ \frac{1}{|\mathcal{N}_{\setminus S}|} \sum_{t \in \mathcal{N}_{\setminus S}} f^{(i)}(v_t,\thetaempiS{i}{S}) \right\}' W_S^{(i)} \left\{ \frac{1}{|\mathcal{N}_{\setminus S}|} \sum_{t \in \mathcal{N}_{\setminus S}} f^{(i)}(v_t,\thetaempiS{i}{S}) \right\},
\end{equation*}
where $\mathcal{N}_{\setminus S} = \{1,\dots,T\} \setminus \mathcal{N}_{S}$.
In ($k,r$)-CV, the averaged validation score of each model
\begin{equation*}
 Q_{\mathrm{valid}}^{(i)} = \frac{1}{\rCk} \sum_{S \subset \{1,2,\dots,r\}: |S| = r-k} Q_{S, \mathrm{valid}}^{(i)}(\thetaempiS{i}{S})
\end{equation*}
is compared, and the model of smaller averaged validation score is selected. The procedure is summarized in Algorithm \ref{alg_cv_gmm}.

\subsection{Consistency of CV in Model Selection}

In this section, we derive the consistency of CV in GMM model selection. Let one of the models is misspecified.
Without loss of generality, we assume the first model is the true model\footnote{Of course, the model selection method should not exploit this fact.}. The true model satisfies the following moment condition:
\begin{align*}
\mathbf{E}[f^{(1)}(v_t,\theta_0^{(1)})]=0.
\end{align*}
The latter model is assumed to be misspecified: that is, for any $\theta^{(2)}$ the following holds:
\begin{align*}
\mathbf{E}[f^{(2)}(v_t,\theta^{(2)})] > 0.
\end{align*}
The misspecification is divided into two local and non-local ones \cite{hall2005}.
\begin{assumption}
The false model is globally misspecified if there exists $\mu(\theta)$ such that $||\mu(\theta)||>0$ and 
\begin{align*}
\inf_{\theta^{(2)} \in \Theta^{(2)}} \mathbf{E}\left[f^{(2)}(v_t,\theta^{(2)})\right] = \mu(\theta).
\end{align*}
\end{assumption}
Alternatively, we can make a weaker assumption that the sample moment of the misspecified model converges to zero slower than that of the true model. This assumption covers cases where the misspecified model is more general (or too general) than the true model. This is the case, for example, the utility function in the true model is a linear function of price but the misspecified model incorporates higher order polynomials.
\begin{assumption}
The false model is said to be locally misspecified if, for every $\epsilon\in(0,1)$, there exists $c_{\epsilon}>0$ such that, when $T$ is sufficiently large, $P[Q_{\mathrm{valid}}^{(1)} < Q_{\mathrm{valid}}^{(2)}]\geq 1-\epsilon$.
\end{assumption}
Note that, in either definition of misspecification, the researcher does not know which model is true, and our interest lies in consistently choosing the true model over a misspecified model based on the dataset.

In the previous literature, Smith (1992) offers a pairwise comparison process for consistent model selection. However, it has some practical disadvantages when applied to empirical research:
(i) A pairwise comparison could be extremely demanding if the space of candidate models is large, and (ii) it may be subject to over-fitting problem. 
To avoid those issues, the most common practice in the field of machine learning is to apply cross-validation (CV) algorithm. In the literature in statistics, Yang (2006,2007) have shown that even the simplest CV procedure can find a true model consistently when the data structure is regression form, i.e. $y_i=f(x_i)+\epsilon_i$.
Likewise to the literature, we define a consistent model selection as below:
\begin{definition}
Assume that model $1$ is correct while model $2$ is wrong in a sense that it is globally misspecified. 
A selection rule is said to be consistent if the probability of selecting model $1$ approaches $1$ as $T\longrightarrow\infty$.
\end{definition}

To derive the consistency of CV, we define the following assumptions.
\begin{assumption}{\rm (strict stationarity)} 
$\mathbf{v}=\{v_t\}$ is a strictly stationary process.
\label{asm_stationary}
\end{assumption}

\begin{assumption}{\rm (regularity condition)} 
Let $f^{(i)}(v_t, \theta)$ and its population analogue $\mathbf{E}[f^{(i)}(v_t, \theta)]$ be continuous on $\theta^{(i)}$ for each $v_t$. Let $\Theta^{(i)}$ be compact and $\mathbf{E}[\sup_{\theta^{(i)} \in \Theta^{(i)}} f^{(i)}(v_t, \theta)]$ be bounded.
\label{asm_regular}
\end{assumption}

\begin{assumption}{\rm (ergodicity)} 
$\mathbf{v}=\{v_t\}$ is an ergodic process.
\label{asm_ergodicity}
\end{assumption}

\begin{assumption}{\rm (identification condition)} 
 Let \begin{equation*}
  \mathbf{E}\left[\frac{\partial f^{(i)}(v_t, \theta^{(i)}_0)}{\partial \theta^{(i)}}\right] 
\end{equation*}
have rank $d$.
\label{asm_identification}
\end{assumption}


In the following we prove the following theorem.
\begin{theorem}
Let Assumptions \ref{asm_stationary}--\ref{asm_identification} hold. Then, $(r,k)$-CV is consistent.
\label{thm_nonlocal_consistency}
\end{theorem}


\subsection*{Proof of Theorem \ref{thm_nonlocal_consistency}}

We first states lemmas that are proven in \cite{hall2005}, and by using them we prove the theorem.
\begin{lemma} {\rm (Consistency of the estimator in the correct model, Theorem 3.1 in \cite{hall2005})}
Let $S \subset \{1,\dots,r\}, |S|=r-k$ be any split in $(k,r)$-CV, and model $1$ be correctly specified.
Let Assumptions \ref{asm_stationary}--\ref{asm_identification} hold. Then,
\begin{equation}
\thetaempiS{1}{S} \rightarrow^{p} \theta_0^{(1)} \label{ineq_consistency}
\end{equation} 
as $T/r \rightarrow \infty$.
\end{lemma}

\begin{lemma} {\rm (Property of a globally misspecified estimator, Theorem 5.2 in \cite{hall2005})}
Let $S \subset \{1,\dots,r\}, |S|=r-k$ be any split in $(k,r)$-CV. 
Let Assumptions \ref{asm_stationary}--\ref{asm_identification} hold. Then, here exists $c>0$ such that
\begin{equation}
Q_0^{(i)}(\thetaempiS{i}{S}) \inprob c
\label{ineq_misspecify_linear}
\end{equation} 
as $T/r \rightarrow \infty$.
\end{lemma}

\begin{lemma} {\rm (Uniform convergence of the moment, Lemma 3.1 in \cite{hall2005})}
Let Assumptions \ref{asm_stationary}--\ref{asm_identification} hold. Then,
\begin{align}
\sup_{\theta^{(1)} \in \Theta^{(1)}} |Q_{S,\mathrm{valid}}^{(i)}(\theta^{(1)})-Q_0^{(1)}(\theta^{(1)})| \inprob 0 \label{ineq_unifg} \\
\sup_{\theta^{(2)} \in \Theta^{(2)}} |Q_{S,\mathrm{valid}}^{(2)}(\theta^{(2)}) - Q_0^{(2)}(\theta^{(2)})| \inprob 0 \label{ineq_unifhtrain}
\end{align}
\end{lemma}

\begin{proof}[Proof of Theorem \ref{thm_nonlocal_consistency}]
We show that, 
\begin{equation}
\sup_{\theta^{(1)} \in \Theta^{(1)}} |Q_{\mathrm{valid}}^{(1)}| \inprob 0
\label{ineq_nlc_one}
\end{equation}
and there exists $c>0$ such that
\begin{equation}
|Q_{\mathrm{valid}}^{(2)}| \inprob c
\label{ineq_nlc_two}
\end{equation}
which imply Theorem \ref{thm_nonlocal_consistency}.
First, 
\begin{align*}
\lefteqn{
|Q_{\mathrm{valid}}^{(1)}-Q_0^{(1)}(\theta_0^{(1)})|
\le \sum_{S \in \{1,\dots,r\}: |S|=r-k} |Q_{S, \mathrm{valid}}^{(1)}(\thetaempiS{1}{S})-Q_0^{(1)}(\theta_0^{(1)})| 
}\nn
&\le \sum_{S \in \{1,\dots,r\}: |S|=r-k} 
\left(
|Q_{S, \mathrm{valid}}^{(1)}(\thetaempiS{1}{S})-Q_0^{(1)}(\thetaempiS{1}{S})|+
|Q_0^{(1)}(\thetaempiS{1}{S})-Q_0^{(1)}(\theta_0^{(1)})| 
\right)
\end{align*}
Inequality \eqref{ineq_unifg} implies the first term converges to zero in probability, and the second term converges to zero in probability by \eqref{ineq_consistency}.
In other words,
\begin{equation}
|Q_{\mathrm{valid}}^{(1)}(\theta^{(1)})-Q_0^{(1)}(\theta_0^{(1)})| \inprob 0
\end{equation}
and by Assumption 3.3 in \cite{hall2005}, 
\begin{equation}
Q_0^{(1)}(\theta_0^{(1)}) = 0
\end{equation}
and thus inequality \eqref{ineq_nlc_one} is derived. 
We next show \eqref{ineq_nlc_two}. We have,
\begin{align*}
\lefteqn{
Q_{\mathrm{valid}}^{(2)}
}\nn
&\ge Q_0^{(2)}(\theta^{(2)}_0) - \frac{1}{{}_r C _k} \sum_{S \in \{1,\dots,r\}: |S|=r-k} \left( |Q_{S,\mathrm{valid}}^{(2)}(\thetaempiS{2}{S}) - Q_0^{(2)}(\thetaempiS{2}{S})| - |Q_0^{(2)}(\thetaempiS{2}{S}) -  Q_0^{(2)}(\theta_0^{(2)})| \right),
\end{align*}
where the first term of the RHS  converges to $c>0$ in probability by \eqref{ineq_consistency}. The second term converge to zero in probability by \eqref{ineq_unifhtrain}. The third term goes to zero in probability by our assumption. Therefore \eqref{ineq_nlc_two} holds.
\end{proof}

\subsection{Statistical testing}

This section proposes a statistical hypothesis testing on our CV-based model selection. 
Let 
\[
 \statNS = \frac{\NSSizehalf(\Qempi{1}-\Qempi{2})}{\sigmaemp^2}
\]
be the test statistics that indicates either the first or the second hypothesis is better than the other. Here, $\sigmaemp^2$ is the estimator of the limiting variance $\sigma_0^2$ of $\statNS$. The null hypothesis of the test is 
\[
\hypnull: \Qtruei{1} = \Qtruei{2}.
\]
These are two alternative hypotheses of interest: The first one indicates $\Model{1}$ is better than $\Model{2}$. That is,
\[
\hypone: \Qtruei{1} < \Qtruei{2}
\]
and the second one indicates $\Model{2}$ is better than $\Model{2}$:
\[
\hyptwo: \Qtruei{1} > \Qtruei{2}.
\]
Following \cite{riversvuong2002,hallpelletier2011}, we discuss conditions where the statistics $\statNS$ is asymptotically normal. We first consider the testing statistics in the general case in Section \ref{sssec_dep}. Moreover, we show in the case the dependency among splits are sufficiently small in Section \ref{sssec_nodep}, where the statistics is represented in a much computationally efficient way.

We pose the following assumption on the structure of the weight matrix that is essentially the same as \cite{hallpelletier2011}:
\begin{assumption}{\rm (parameterization of the weight matrix)}
Let $\Wtrue{i}$ depends on a vector nuisance parameter $\tautrue{i}$ and $\tauemp{i}$ is the estimator of $\tautrue{i}$ as $\Wtrue{i} = \Wtrue{i}(\tautrue{i})$ and $\Wemp{i} = \Wemp{i}(\tauemp{i})$.  It is assumed that the nuisance parameter satisfies 
\[
 \NSSizehalf (\tauemp{i}-\tautrue{i}) = - \Astari{i} \NSSizehalfinv \sum_{t \in \NS} Y_t^{(i)} + o_p(1)
\]
for some symmetric matrix of constants $\Astari{i}$ and data-dependent vector $Y_t^{(i)}$, and the weight matrix satisfies
\[
 \NSSizehalf \left( \vech[\Wemp{i}] - \vech[\Wtrue{i}] \right) = \Deltai{i} \NSSizehalf (\tauemp{i}-\tautrue{i}) + o_p(1)
\]
for some matrix of constants $\Deltai{i}$. 
\end{assumption}

To discuss statistical testing, we need to have asymptotic normality property. The following assumption guarantees that the moment is ``well-behaved'' around the optimal value $\thetatruei{i}$.
\begin{assumption}{\rm (regularity condition on derivative)} 
\begin{itemize}
\item The derivatie matrix $\partial f^{(i)}(v, \theta^{(i)})/\partial {\theta^{(i)}}'$ exists and is continuous on $\Theta^{(i)}$ for each $v$.
\item $\thetatruei{i}$ lies in the interior of $\Theta^{(i)}$.
\item $\Expect[\partial f^{(i)}(v, \thetatruei{i})/\partial {\theta^{(i)}}']$ exists and is finite.
\item $\Expect[\partial f^{(i)}(v, \theta^{(i)})/\partial {\theta^{(i)}}']$ continuous on some neighborhood $N_\epsilon$ of $\thetatruei{i}$.
\item $\sup_{\theta^{(i)} \in N_\epsilon} ||(1/T)\sum_{t=1}^T\partial f^{(i)}(v, \theta^{(i)})/\partial {\theta^{(i)}}' - \Expect[\partial f^{(i)}(v, \theta^{(i)})/\partial {\theta^{(i)}}']|| \inprob 0$.
\end{itemize}
\label{asm_regderiv}
\end{assumption}

For the ease of discussion, we further add the following notation.
Let $\FSi{i}{S} = \NmSSizehalfinv \sum_{t \in \NmS} \{f^{(i)}(v_t,\theta^{(i)})\}$. 
Let $\Gtruei{i} = \Expect[\partial f^{(i)}(v_t,\theta^{(i)})/\partial \theta^{(i)}]$, and its empirical counterpart be $\Gempi{i} = \NSSizeinv \sum_{t \in \NmS} (\partial f^{(i)}(v_t,\theta^{(i)})/\partial \theta^{(i)})$.
Let $\{S_1,\dots,S_{\rCk}\} = \combinationSetS$ be the set of all splits.
We also denote $\theta = (\theta^{(1)}, \theta^{(2)})$, and $\thetatrue$ and $\thetaempS{S}$ are defined in the same way.

\subsubsection{General splitting}
\label{sssec_dep}

Then, $\Vtrue$ is the 
\[
\Vtrue = \left(\begin{array}{cccc}
\myC{1, 1} & \myC{1, 2} & \cdots & \myC{1, \rCk} \\
\myC{2, 1} & \myC{2, 2} & \cdots & \myC{2, \rCk} \\
\vdots & \vdots & \ddots & \vdots \\
\myC{\rCk, 1} & \myC{\rCk, 2} & \cdots & \myC{\rCk, \rCk}
\end{array}
\right)
\]
where $\myC{j, j'}$ is a submatrix such that
\begin{align*}
\myC{j, j'} &= \lim_{T \rightarrow \infty} \Cov\left(\sum_{t \in \NmSj{j}} \xi_t(\thetatrue), \sum_{t' \in \NmSj{j'}} \xi_{t'}(\thetatrue) \right) \nn
\xi_t(\theta)      &= 
\left[
f^{(1)}(v_t,\thetai{1})'-\Expect[f^{(1)}(v_t,\thetai{1}]',Y_t^{(1)'},f^{(2)}(v_t,\thetai{2})'-\Expect[f^{(2)}(v_t,\thetai{2})]',Y_t^{(2)'}
\right]'.
\end{align*}
Moreover, 
\begin{align*}
\Rtrue &= \left[ \Rtruei{1}, \Rtruei{2}, \Rtruei{1}, \Rtruei{2}, \dots , \Rtruei{1}, \Rtruei{2}\right]' \nn
\Rtruei{i} &= \left[
    \begin{array}{c}
      2 \Wtrue{i} \Expect[f^{(1)}(v_t,\theta^{(1)})] \\
      -  \Astari{i} {\Deltai{i}}' B_i' \Expect[f^{(1)}(v_t,\theta^{(1)})] \otimes \Expect[f^{(1)}(v_t,\theta^{(1)})]
    \end{array}
    \right]
\end{align*}
where $B_i$ is the $\dimf{i}^2 \times \dimf{i}(\dimf{i}+1)/2$ matrix such that $\myvec(\Wtrue{i}) = B_i \vech(\Wtrue{i})$.

\begin{assumption}
\begin{enumerate}[label=(\alph*)]
\item 
 Assume that 
$[\FSi{1}{S_1}, \FSi{2}{S_1}, \FSi{1}{S_2}, \FSi{2}{S_2},\dots,\FSi{1}{S_{\rCk}}, \FSi{2}{S_{\rCk}}] \rightarrow N(0, \CovS(\theta))$.
Where $\CovS(\theta)$ is a positive semi-definite matrix of constants. 
\item rank$\{\Gtruei{i}\} = d$.
\item $\Shalf (\thetaempiS{i}{S} - \theta^{(i)}) = O_p(1)$.
\item The empirical estimator of each $\CovS(\theta)$ converges as $\CovSemp(\thetaempS{S}) \rightarrow \CovS(\thetatrue)$. \end{enumerate}
\label{asm_asymnormaltest}
\end{assumption}

\begin{theorem}{\rm (asymptotic normality of $\statNS$)}
Assume that both models $\Model{1}$ and $\Model{2}$ are misspecified. Assume that Assumption \ref{asm_asymnormaltest} holds. Assume Assumptions \ref{asm_stationary}, \ref{asm_regular}, \ref{asm_ergodicity}, and \ref{asm_regderiv} hold. Assume that the null hypothesis $\hypnull$ holds. Let $\Wtrue{i} = I_{\dimf{i}}$. Then, 
\[
\statNS \rightarrow N(0,1).
\]
\label{thm_hyp_asymnormal}
\end{theorem}

\begin{remark}
Theorem \ref{thm_hyp_asymnormal} poses the assumption that both models are misspecified. As discussed in \cite{hallpelletier2011}, this assumption is essential: One can check that, under correctly specified models, the distribution of $\statNS$ does not have asymptotic normality. 
\end{remark}

\begin{remark}
As discussed in \cite{HallInoue2003} a constant weight matrix has the best rate of convergence in misspecified models and thus the assumption of identity $\Wtrue{i}$ in Theorem \ref{thm_hyp_asymnormal} is reasonable. 
\end{remark}

\begin{proof}[Proof of Theorem \ref{thm_hyp_asymnormal}]
The theorem is an extension of Theorem 1 in \cite{hallpelletier2011} to multiple splitting. 
The mean value theorem applied to $\QempvalidSi{i}{\thetaempiS{i}{S}}$ around $\thetatruei{i}$, we obtain 
\[
 \QempvalidSi{i}{\thetaempiS{i}{S}} = \QempvalidSi{i}{\thetatruei{i}} + \left\{ \frac{\partial \QempvalidSi{i}{\thetai{i}}}{\partial \thetai{i}} \Biggr|_{\thetai{i} = \thetabariS{i}{S}} \right\}' (\thetaempiS{i}{S} - \thetatruei{i})
\]
where $\thetabariS{i}{S} = \lambda_S \thetatruei{i} + (1 - \lambda_S) \thetaempiS{i}{S}$ for some $\lambda_S \in [0,1]$.
Let 
\[
\Phii{i}(\thetatruei{i}) = 2 \Gtruei{i}(\thetatruei{i})' \Wtrue{i} \Expect[f^{(i)}(v_t,\thetatruei{i})].
\]
From our assumptions, we obtain
\[
 \QempvalidSi{i}{\thetaempiS{i}{S}} = \QempvalidSi{i}{\thetatruei{i}} + \left\{ \frac{\partial \thetaempiS{i}{S}}{\partial \thetai{i}} \right\}' (\thetaempiS{i}{S} - \thetatruei{i}) + o_p(\NmSSizehalfinv),
\]
and thus 
\begin{align}
\NmSSizehalf [\QempvalidSi{1}{\thetaempiS{1}{S}} - \QempvalidSi{2}{\thetaempiS{2}{S}}]
&= \NmSSizehalf [\QempvalidSi{1}{\thetatruei{1}} - \QempvalidSi{2}{\thetatruei{2}}] \nn
&\ +\left\{ \Phii{1}(\thetatruei{1}) \right\}' \Shalf (\thetaempiS{1}{S} - \thetatruei{1}) \nn
&\ -\left\{ \Phii{2}(\thetatruei{2}) \right\}' \Shalf (\thetaempiS{2}{S} - \thetatruei{2}) \nn
&\ +o_p(1).
\label{ineq_asymnorm_allterms}
\end{align}
Note that the GMM estimator minimizes the moment condition, which implies $\Gempi{i}(\thetaempiS{i}{S})' \Wemp{i} (1/\NmSSizeinv) \sum_{t \in \NmS} f^{(i)}(v_t,\thetaempiS{i}{S}) = 0$. This fact implies the third and fourth terms of \eqref{ineq_asymnorm_allterms} vanishes. Namely,
\begin{align}
\NmSSizehalf [\QempvalidSi{1}{\thetaempiS{1}{S}} - \QempvalidSi{2}{\thetaempiS{2}{S}}]
&= \NmSSizehalf [\QempvalidSi{1}{\thetai{1}} - \QempvalidSi{2}{\thetai{2}}] \nn
&\ +o_p(1).
\label{ineq_asymnorm_allterms}
\end{align}
With the choice $\Wtrue{i} = I_{\dimf{i}}$ for the weighting matrix, and by using the symmetry of the moment we obtain
\begin{multline}
\NmSSizehalf [\QempvalidSi{1}{\thetaempiS{1}{S}} - \QempvalidSi{2}{\thetaempiS{2}{S}}]
= \nn 
2\Biggl\{ 
\mui{1}(\thetatruei{1}) \NmSSizehalfinv \sum_{t \in \NmS} [f^{(1)}(v_t,\thetatruei{1})-\mui{1}(\thetatruei{1})]  \nn
- \mui{2}(\thetatruei{2}) \NmSSizehalfinv \sum_{t \in \NmS} [f^{(2)}(v_t,\thetatruei{2})-\mui{2}(\thetatruei{2})]
\Biggr\} + o_p(1),
\end{multline}
which, combined with our assumptions, completes the proof.
\end{proof}

\subsubsection{When dependency among validation splits is small}
\label{sssec_nodep}

Calculating the asymptotic variance of Theorem \ref{thm_hyp_asymnormal} requires a calculation of a matrix with its size proportional to the number of splits, which in some cases is computationally prohibitive.   
This section consider the case where the dependency between the validation data is  sufficiently small. In such a case, we can circumvent the computation of a large matrix.

In particular, the leave-one-out CV (special case of our CV with $k=1$) when each datapoint is identically and independently distributed (i.i.d), the following assumption holds:
\begin{assumption}
Assume that each validation split $\{\NmSj{j}\}$ is independent and identically distributed.
\label{asm_asymnormaltest_loo}
\end{assumption}

\begin{theorem}{\rm (asymptotic normality of $\statNS$, Leave-one-out)}
\label{thm_hyp_asymnormal_loo}
Let assumptions in Theorem \ref{thm_hyp_asymnormal} hold. Let Assumption \ref{asm_asymnormaltest_loo} holds.
Then, the limit variance is written as
\[
\sigma^2 = \sum_{S \in \combinationSetS} \left( \Rtruesingle' \Vtruesingle(S) \Rtruesingle \right)
\] 
where
\begin{align}
 \Rtruesingle &= \left[ \Rtruei{1}, \Rtruei{2}\right]' \nn
\Rtruesinglei{i} &= \left[
    \begin{array}{c}
      2 \Wtrue{i} \Expect[f^{(1)}(v_t,\theta^{(1)})] \\
      -  \Astari{i} {\Deltai{i}}' B_i' \Expect[f^{(1)}(v_t,\theta^{(1)})] \otimes \Expect[f^{(1)}(v_t,\theta^{(1)})]
    \end{array}
    \right] \nn
 \Vtruesingle(S) &= \lim_{T \rightarrow \infty} \Var(\sum_{t \in \NmS} \xi_t)
\end{align}
And The asymptotic normality holds:
\[
\statNS \rightarrow N(0,1).
\]
\end{theorem}
The proof of Theorem \ref{asm_asymnormaltest_loo} directly follows by following the same steps as Theorem \ref{thm_hyp_asymnormal} with additional fact that Assumption \eqref{asm_asymnormaltest_loo} implies the block-diagonal property of $\Vtrue$ as $\myC{i,j} \rightarrow 0$ for $i \ne j$ and the identity of each block.

\section{MONTE-CARLO EXPERIMENTS IN LINEAR MODEL}\label{3-simulation1-regression}
In this section we present a simple simulation of instrumental variables (IV) regression models to illustrate the consistency of our cross-validation algorithm of model selection. This example also highlights how GMM-minimand-based model comparison and cross-validation can exhibit different results. The setting is similar to the one on \cite{hallpelletier2007}. Suppose the true data generating process is
\begin{align*}
\boldsymbol{y} =  X_{1} \boldsymbol{\beta^1} + X_{2} \boldsymbol{\beta^2}+ Z_2 \boldsymbol{\alpha} +\boldsymbol{\epsilon},
\end{align*}
where 
$\boldsymbol{y}$ is a $T\times 1$ vector and $X_{1}$ and $X_{2}$ are $T\times p_1$ and $T\times p_2$ matrix respectively. $X_1$ and $X_2$ are generated from instrumental variables as
\begin{align*}
X_1 &= Z_1 \boldsymbol{\delta^1} +\boldsymbol{\xi}^1,\\
X_2 &= Z_2 \boldsymbol{\delta^2} +\boldsymbol{\xi}^2,
\end{align*}
where $Z_1$ and $Z_2$ are $T\times c_1$ and $T\times c_2$ matrix respectively.

We consider a case where we have two candidate models to compare. The first model exploits the explanatory variables $X_1$ and instrumental variables $Z_1$.  
\begin{align*}
\mathcal{M}^{1}: &\boldsymbol{y}=X_1 \boldsymbol{\beta}+\boldsymbol{\epsilon}^1,\\ &\mathbb{E}[Z_1' \boldsymbol{\epsilon}^1 ]=0,
\end{align*}
whereas the second model employs $X_2$ and $Z_2$;
\begin{align*}
\mathcal{M}^{2}: &\boldsymbol{y}=X_2 \boldsymbol{\beta}+\boldsymbol{\epsilon}^2,\\ &\mathbb{E}[Z_2' \boldsymbol{\epsilon}^2 ]=0.
\end{align*}

Each model has different explanatory variables as well as the set of instrumental variables so that two models are non-nested. In addition, there are two important differences between the two candidates. First, the second model can be "misspecified" when $\alpha\neq 0$, because the instrumental variables $Z_2$ influences $\boldsymbol y$ directly and thus IVs are not independent from $\boldsymbol \epsilon^2$. When $\alpha>0$ and does not decrease with the number of observations, i.e. $\alpha=10$, it is globally misspecified, which results in inconsistent estimates of the parameters.

The second difference is that the number of the variables. In the following, we assume that $p^1\leq p^2$, meaning that the second model has a larger number of explanatory variables. As discussed earlier, this may cause "over-fitting" issue to the estimation even if the model is falsely specified. In such a case, previous literature proposes the ways to penalize the model by the number of parameters (\cite{andrews1999}). We compare the performance of the proposed method with the ones of those existing methods in the later section.

Though this example may seem to be somewhat arbitrary, similar problems arise in many situations when econometric models are compared. Specifically, one model can be flexible (or even "over flexible") but misspecified, while the other is simpler but accurate. Some researchers may not value the simplicity, but they would prefer a "correctly specified" model than misspecified models. For example, think of a case where economists try to explain wage from education and other variables, where education is endogenous and has to be proxied by IVs. The misspecified model includes incorrect IVs that gives bias to the estimate of the coefficient. Even if one model exhibits a good fit to the data, if the coefficient of interest is not properly estimated, such a model does not serve well for labor economists. In those occasions, our algorithm serves to help researchers to find the most "correct" model. Our method is general enough so that any specification can be compared. 

\subsection{Results}
First we consider the case where over-fitting is a concern as the misspecified model has more parameters therefore could exhibit better fit to the data. We compare our methodology in this case to the model selection procedures proposed by \cite{andrews1999} as well as simple GMM comparison as in the previous section. \cite{andrews1999} defines GMM-AIC and GMM-BIC criterion as
\begin{align*}
\text{GMM-AIC: } &T Q^{(i)}_T (\theta_T^{i})-2(|c^{i}|-p^{i});\\
\text{GMM-BIC: } &T Q^{(i)}_T (\theta_T^{i})-(|c^{i}|-p^{i})ln T,
\end{align*}
for $i=1,2$. The procedure chooses the model that exhibits smaller value of the criterion.

Figure \ref{fig: overfit} shows the empirical probability of choosing the correctly specified model by cross-validation. One can see that, even when the model 2 has larger number of variables, it chooses the model 1 with very high chance even when the data is limited. When the bias parameter of the model 2 $\alpha$ is as large as $12.$, it selects the first model with probability $91.2\%$ even when the data size is only $100$ and the second model has $9$ variables compared to $3$ of the first model.
 
On the other hand, GMM based model selection performs extremely poorly when the misspecified model has much more variables than the first model. When $p^2=9$, even with data size $1600$ the accuracy is as bad as $59.1\%$, only slightly above chance level of $50\%$ (when $\alpha=12.$). With data size $200$, it chooses the second model only for $15.7\%$, clearly indicating it is subject to over-fitting.

Note that in our setting, GMM-AIC and GMM-BIC exhibit exactly same choice of models as simple GMM based selection. This is due to the unbalance of two terms in the criterion. In our case, the first term is typically on order of more than $10^5$, while the second term is no greater than $10^2$. Many factors influence the magnitude of the first term, such as the choice of weighting matrix or number of moment conditions. Our result suggests that while cross validation robustly performs in many situations, performance of GMM based model selection is sensitive to those settings.
 
We turn to the case where the two models have the same number of parameters, while the second model is misspecified. As the number of parameters is the same across two models, note that GMM, GMM-BIC, and GMM-AIC simply choose the model with smaller GMM minimand. Figure \ref{fig: p1eqp2_global} compares the performance of cross-validation algorithm and the GMM minimand based model selection when the second model is globally misspecified. The $y$-axis shows the probability that the correctly specified model is chosen by each algorithm. The result indicates that when overfitting is not a concern, GMM based model selection performs slightly better than cross validation, especially when the data is smaller.

\section{NONLINEAR EXPERIMENT: COLLUSION DETECTION}\label{4-simulation2-collusion}
In this section, we demonstrate another Monte-Carlo study to show how our algorithm works in a structural estimation incorporating nonlinear and non-nested models. Specifically, we simulate and estimate a variant of a price collusion model suggested by Bresnahan (1987). The goal of our model selection procedure is to detect whether the firms are colluded, or determining the price competitively using the share and price data. The underlying idea is that the prices of the products of colluded firms are determined to maximize the joint profit, while the competitive price should maximizes the profit of individual firms. Therefore, given the same (true) parameters in demand and cost function, the pricing pattern varies according to the collusive structure. A methodology to study whether collusive behavior exists within a certain industry is by itself an important research topic because ignoring the possibility of collusive pricing may lead to a biased inference of cost estimation, which could be a critical problem for policy implication in applications such as merger analysis.

In the same way as the previous section, we compare the performance of CV-based algorithm to GMM-minimand-based algorithm based on the theory of \cite{riversvuong2002}. Note that since the number of parameters in a model does not vary across collusive structure,  AIC or BIC adjustment does not influence the model selection criteria. We show that in a realistic sample size, CV performs better than in-sample comparison in many cases. 

The shares and prices are simulated from a standard logit demand system and static pricing. We simulate data assuming a certain collusive structure. Then we test if and how often CV algorithm can discover the assumed collusive structure. The estimation process is similar to \cite{hu2014collusion}. 

\subsection{Model}
Assume each firm produces a single product and denote them as $j=1,...,J$. The markets are denoted as $t=1,...,T$. The demand is assumed to be a simple logit demand specification: the utility of a consumer $i$ purchasing a product $j$ in a market $t$ is expressed as 
\begin{align*}
u_{ijt}=X_{jt} \beta+\alpha p_{jt}+\xi_{jt}+\epsilon_{ijt},
\end{align*}
where $X_{jt}$ is the observed characteristics that influence the demand and $\xi_{jt}$ is the unobserved utility shock . Assuming $\epsilon_{ijt}$ follows i.i.d type-I extreme value distribution, the share function is 
\begin{align*}
D_{jt}(\mathbf{p}_t)=\frac{\exp(X_{j} \beta+\alpha p_{jt}+\xi_{jt})}{\sum_{j'=1}^J\exp(X_{j't} \beta+\alpha p_{j't}+\xi_{j't})}M_t,
\end{align*}
where $\mathbf{p}_t=\{p_{jt}\}_{j=1,...,J}$ is the vectorized prices and $M_t$ is the market size which is known to the researcher. For simplicity, we do not allow random-coefficients (\cite{berry1995automobile}) as typically done in applications. 

Firms' marginal cost is expressed as
\begin{align*}
MC_{jt}=Y_{jt} \gamma +\lambda_{jt}
\end{align*}
,where $Y_{jt}$ is the observed characteristics that affect the marginal cost, and $\lambda_{jt}$ is the i.i.d cost shocks. The profit of each product is
\begin{align*}
\pi_{jt}(\mathbf{p}_t)=(p_{jt}-MC_{jt})D_{jt}(\mathbf{p}_t).
\end{align*}

We assume that colluded firms jointly maximize their net profit, sum of $\pi_{jt}$ over $j$ in a group. Define $\boldsymbol{\Delta}$  as a $J\times J$ matrix of price elasticity of colluded products where the $(j,r)$th element is
\begin{align*}
\boldsymbol{\Delta}_{jr}=
\begin{cases}
-\pdif{D_r}{p_j} & \text{if $j$ and $r$ are colluded}  \\ 
0 & \text{otherwise} .
\end{cases}
\end{align*}
By solving the first order conditions, the equilibrium prices are determined to satisfy
\begin{align*}
\mathbf{p}_t = \left( \boldsymbol{\Delta} \right)^{-1} \mathbf{D}_t - \mathbf{MC}_t,
\end{align*}
where $\mathbf{D}_t$ and $\mathbf{MC}_t$  are a vectorized representation of $D_{jt}(\mathbf{P}_t)$ and $\{ MC_{jt} \}_{j=1,..,J}$ respectively.

\subsection{Estimation and Model Selection}
The parameter estimation under each model follows a standard GMM procedure with instrumental variables. Let $Z$ be instrumental variables that influence the price but are not correlated with the unobserved shocks $\xi$ and $\lambda$. Given a model, the parameters are chosen to minimize the GMM objective defined from the moment condition
\begin{align*}
\Expect [\xi Z]=0\\
\Expect [\lambda Z]=0.
\end{align*}

The instrumental variables $Z$ include (i) own characteristics, (ii) square of own characteristics, (iii) mean of characteristics in a market, and (iv) square of mean characteristics in a market. The weighting matrix is set to be $W=(Z'Z)^{-1}$.

The candidate models are represented as partitions of firms into price-colluded groups. For instance, if the number of firms is two ($j=1,2$), the possible models are either competitive ($\{\{1\},\{2\}\}$) or collusive ($\{\{1,2\}\}$). If three firms (j=1,2,3), possible models are $\{ \{1\},\{2\},\{3\} \}$ (all competitive), $\{ \{1\}, \{2,3\} \}$, $\{ \{1,2\}, \{3\} \}$, $\{ \{1, 3\}, \{2\} \}$, and $\{ \{1,2,3\} \}$ (all colluded).

\subsection{Simulation Results}

We consider different number of observed markets, $T=\{25,50,75,100\}$, realistic numbers for real world application\footnote{For instance, \cite{nevo2001measuring} observes $94$ independent markets.}. We also vary the true value of price coefficient to test the performance with different difficulty of model selection. Along with the data size, the difficulty of model selection depends on how different the observed data would be across different models. In this particular example, the key difference between models is generated from cross price elasticity. When the cross price elasticity is low, competitive price and colluded price do not differ as much, which makes it harder to find the true model. In logit-demand, the cross price elasticity is calculated by multiplying the share of the two products. Thus, lower price coefficient generally makes model selection easer as it increases the realized share, and the cross price elasticity as a result. For each setting, we generate 100 synthetic dataset and perform the model selection in each.

Table \ref{tab:CD_Score_CV} reports the mean and standard deviation of CV score across true models and candidate models with the price coefficient equals to $-.1$ and $-.3$. The second column represents the true partition of firms, and the third to seventh are the results corresponding to each candidate model. The CV score of the true model is on average smaller than the mis-specified models in any specification. Also, the standard deviation of the score is smaller for the true model. Both mean and standard deviation of the true model decline in the number of observations.

We report the probability that each candidate model is chosen by our algorithm in table \ref{tab:CD_ChoiceProb_CV}. In each setting, the probability to find the true model increases in the number of markets, which corresponds to our theoretical finding in section 2. For comparison, Table \ref{tab:CD_ChoiceProb_GMM} presents the same for GMM-minimiand comparison. 

Figure \ref{fig:CD_graph} compares the performance of our model selection to a simple in-sample GMM fit comparison under different price coefficient. It shows that our CV algorithm generally performs better than in-sample comparison. The difference is particularly large when the true model is partially colluded (second column). As seen in Table \ref{tab:CD_ChoiceProb_GMM}, GMM comparison tends to select all-competitive model in such a case. 

\section{CROSS-VALIDATION APPROACH TO MPEC Estimation}\label{5-MPEC-theory}
In this section, we propose a method to apply cross-validation algorithm when estimation is based on Mathematical Programming of Equilibrium Constraint (MPEC)  approach proposed by \cite{su2012constrained}. MPEC approach formulates the estimation as an optimization problem with constraints: The variables of the optimization consists of structural parameters as well as endogenous latent economic variables, and the constraints among the variables represent the equilibrium condition that the economic model requires. 

The application of the cross validation procedure to MPEC estimation is not straightforward: If parameters estimated from training data is substituted in a MPEC model with test data directly, the constraints would be not satisfied in general. In such a case, we cannot directly compare GMM objective on test data across models since we also have to consider the violation of constraints as indication of model misfit.

Taking the above discussion into consideration, we propose a modified cross validation procedure. We differentiate the choice variables for the optimization problem into two categories: model variables and observation-specific variables. Model variables are specific to the model, therefore shared across training and test data. Observation-specific variables are latent variables defined on each observation. For instance, in BLP demand estimation example on \cite{dube2012improving}, the price elasticity is a parameter assumed to be constant across observations, thus treated as a model variable. Meanwhile, the unobserved utility shock ($\xi_{jt}$ in their notation) is defined for each datapoint, thus regarded as observation specific.

Our modification is simple. In training data, we jointly choose the model variables and observation-specific variables to optimize the GMM objective function with equilibrium constraints. In test data, we still solve a constrained optimization problem, but only with respect to observation-specific variables while the model variables are set to the estimates from training data. The algorithm is described in detail below and summarized in Algorithm \ref{alg_cv_gmm_mpec}.

\subsection{GMM-MPEC}
We first outline the MPEC formulation of parameter estimation. Here we follow the notation of \cite{su2012constrained} except that we allow some endogenous variables to be observation-specific. Suppose an econometric model $\mathcal{M}_i$ is expressed with the parameter vector $\theta$, a vector of endogenous variables $\sigma$, and  endogenous variables that are observation-specific $ \eta$, and the equilibrium constraint $h(\theta,\sigma,\eta)=0$. In MPEC formulation, each model is characterized by a set of moment conditions with equilibrium constraints:
\begin{align*}
\mathcal{M}_i \Rightarrow &E[f^{(i)}(v_t,\theta_0^{(i)},\sigma_0^{(i)},\eta_{0}^{(i)})]=0\\
&\text{s.t.}\\
&h^{(i)}(\theta_0^{(i)},\sigma_0^{(i)},\eta_{0}^{(i)})=0.
\end{align*}

Given the observation $\{v_t\}_{t=1,,,T}$, the parameters of each model are estimated via MPEC:
\begin{align}
(\theta_T^{(i)},\sigma_T^{(i)},\eta_T^{(i)})=&\argmin_{\theta^{(i)},\sigma^{(i)},\eta^{(i)} } Q_T^{(i)}(\theta^{(i)},\sigma^{(i)},\eta^{(i)})\\
&\text{s.t.}\\
&h^{(i)}(\theta^{(i)},\sigma^{(i)},\eta^{(i)})=0.
\label{ineq_gmmmpecargmin}
\end{align}
where 
\begin{align*}
&Q_T^{(i)}(\theta^{(i)},\sigma^{(i)},\eta^{(i)})\\&=\left\{ \frac{1}{T} \sum_{t=1}^T f^{(i)}(v_t,\theta^{(i)},\sigma^{(i)},\eta^{(i)}) \right\}' W_T^{(i)} \left\{ \frac{1}{T} \sum_{t=1}^T f^{(i)}(v_t,\theta^{(i)},\sigma^{(i)},\eta^{(i)}) \right\}. 
\end{align*}

Let $\theta_T^{(i), \GMMMPEC}$ be the parameters that are solution of Eq. \eqref{ineq_gmmmpecargmin}, and let $\theta_T^{(i), \GMM}$ be the solution of standard GMM (i.e., Eq. \eqref{eq_gmmargmin}). Moreover, let 
 \begin{align}
V_T^{(i), \GMMMPEC}(\theta) =
&\min_{\sigma^{(i)},\eta^{(i)} } Q_T^{(i)}(\theta, \sigma^{(i)},\eta^{(i)})\\
&\text{s.t.}\\
&h^{(i)}(\theta, \sigma^{(i)},\eta^{(i)})=0,
\end{align}
and $V_T^{(i), \GMM}(\theta) = Q_T^{(i)}(\theta)$.
The equivalence of GMM and GMM-MPEC implies 
\begin{align}
\theta_T^{(i), \GMMMPEC} &= \theta_T^{(i), \GMM} \nn
V_T^{(i), \GMMMPEC}(\theta) &= V_T^{(i), \GMM}(\theta).
\label{eq_gmm_mpec_equiv}
\end{align}

\subsection{Cross-Validation in GMM-MPEC Approach}

We split the observations in the same way as section 2.
The moment on the datapoints $S$ is
\begin{align*}
&Q_S^{(i)}(\theta^{(i)},\sigma^{(i)},\eta^{(i)})\\
&=\left\{ \frac{1}{|\mathcal{N}_S|} \sum_{t \in \mathcal{N}_S} f^{(i)}(v_t,\theta^{(i)},\sigma^{(i)},\eta^{(i)}) \right\}' W_S^{(i)} \left\{ \frac{1}{|\mathcal{N}_S|} \sum_{t \in \mathcal{N}_S} f^{(i)}(v_t,\theta^{(i)},\sigma^{(i)},\eta^{(i)}) \right\}.
\end{align*}
We train the model to minimize the moment under equilibrium constraint. The trained model is denoted as
\begin{align*}
(\theta_S^{(i)},\sigma_S^{(i)},\eta_S^{(i)})=&\argmin_{\theta^{(i)},\sigma^{(i)},\eta^{(i)} } Q_S^{(i)}(\theta^{(i)})\\
&\text{s.t.}\\
&h^{(i)}(\theta^{(i)},\sigma^{(i)},\eta^{(i)})=0.
\end{align*}

Once the model is trained, it is validated by the rest of datapoints. 
Instead of simply evaluating the GMM objective in the validation data at the trained model parameters, observation-specific endogenous variables need to be chosen so that the equilibrium constraints are satisfied. We do so by minimizing the GMM objective subject to equilibrium constraints with respect to $\eta$ only, while model parameters are fixed at trained value. Formally,
\begin{align*}
&Q_{S,\mathrm{valid}}^{(i)}=\\
&\argmin_{\eta^{(i)}}
\left\{ \frac{1}{|\mathcal{N}_{\setminus S}|} \sum_{t \in \mathcal{N}_{\setminus S}} f^{(i)}(v_t,\theta_S^{(i)},\sigma_S^{(i)},\eta^{(i)}) \right\}' W_S^{(i)} \left\{ \frac{1}{|\mathcal{N}_{\setminus S}|} \sum_{t \in \mathcal{N}_{\setminus S}} f^{(i)}(v_t,\theta_S^{(i)},\sigma_S^{(i)},\eta^{(i)}) \right\}\\
&\text{s.t.}\\
&h^{(i)}(\theta_S^{(i)},\sigma_S^{(i)},\eta^{(i)})=0.
\end{align*}

The averaged validation score of each model
\begin{align*}
 Q_{\mathrm{valid}}^{(i)} = \frac{1}{{}_r C _k} \sum_{S \subset \{1,2,\dots,r\}: |S| = r-k} Q_{S, \mathrm{valid}}^{(i)}
\end{align*}
is compared and the model of smaller averaged validation score is selected. 

\begin{remark}{\rm (consistency of GMM-MPEC)}
From \eqref{eq_gmm_mpec_equiv} and Theorem \ref{thm_nonlocal_consistency}, the consistency of GMM-MPEC with the same assumption on the moment directly follows.
\end{remark}

\section{APPLICATION: DYNAMIC DEMAND AND DYNAMIC PRICING MODEL ON ONLINE RETAILER DATA}\label{6-simulation3-dynamicpricing}

In this section, we perform our model selection procedure in a structural model with a real-world dataset. The models we compare are dynamic and static demand and pricing model that are taken from \cite{conlon2012dynamic}. In particular, we first apply our cross-validation algorithm to test either the state-of-the-art dynamic demand model (\cite{gowrisankaran2012dynamics}) or the traditional static demand model (\cite{berry1995automobile}) has stronger explanatory power in the consumer behavior. To this aim, we use monthly sales and price data of an online-retail shop. Furthermore, we consider supply side dynamics of pricing that takes the seasonality and consumer skimming into consideration such as \cite{nair2007intertemporal}: We investigate whether or not such a model explains the observed pricing pattern better than traditional static profit maximization model that is based on the consumer model selected in the previous step.

Structural estimation of a dynamic model has been an important frontier in industrial organization, both on demand side and supply side. On demand side, dynamic model of consumer behavior has been widely applied by researchers recently (\cite{gowrisankaran2012dynamics}). The underlying idea in the dynamic demand model is that consumers are forward-looking regarding the changes in the market such as price and make a dynamic decision by considering the future market state. Such a model is justified by the fact that important parameters such as price elasticity could be severely mis-estimated by ignoring the forward-looking behavior of consumers. Meanwhile, similar mis-estimation would occur if a researcher applies a dynamic model in the case the consumers are in fact myopic. From a market level data, it is not directly visible if consumers are forward-looking or myopic. 


Contrary to the demand side, dynamic pricing in supply side has a long history of theoretical studies dating back to \cite{coase1972durability}. Nevertheless, little empirical attention is paid until recent years (\cite{nair2007intertemporal}). Under certain conditions, firms have the incentive to determine current price by taking its effect on the future profit into consideration. For example, when consumers are heterogeneous in an evaluation of a product, firms are motivated to "skim" high-evaluation consumers in earlier periods by setting a high price and later lower it. With myopic consumers (as in \cite{luo2015operating}), the pricing decision boils down to a dynamic programming of a firm in the case of monopoly or a dynamic game between firms in the case of oligopoly. If the consumers are also forward-looking, the pricing boils down to a dynamic game between consumers and firms as studied in \cite{nair2007intertemporal}. In this case, the observed price and demand are interpreted as a result of dynamic equilibrium. 

It is not straightforward to infer if the pricing is dynamic or not from the market level data. A declining tendency on the price does not always indicate that firms are making pricing decision dynamically: If the consumers are heterogeneous in either product evaluation or price sensitivity and leave market after purchase, a myopic optimal price may be decreasing in periods since the remaining consumers are more price elastic. 

Applying dynamic pricing model to data generated from myopic pricing would cause a significant bias in the estimates of supply-side parameters such as marginal cost. For instance, a dynamic pricing model may interpret an observed high price in a certain period as a firm sparing some demand for the future, while it is a result of high marginal cost in truth. Therefore, estimation of supply-side model parameters such as marginal cost requires researchers to know if firms are myopic or forward-looking. 

As it is important to correctly specify the dynamic feature of the agent's decision making both on demand and supply side, researchers are encouraged to verify whether the decision making is static or dynamic from the data rather than appealing to intuition, desirably based on real-world datasets. Regarding this aspect, we demonstrate our cross-validation algorithm to compare two by two alternative models; dynamic or myopic consumers, and dynamic or myopic firms. The models are estimated via GMM-MPEC.  We take a simple dynamic model from \cite{conlon2012dynamic}. 

We perform estimation and model selection on a dataset of price and sales of an online-retailer based in UK. The data is taken from the University of California Irvine (UCI) Machine Learning Repository (henceforth, UCI). UCI repository consists of more than 300 datasets. The data used in this study is available here at \url{https://archive.ics.uci.edu/ml/datasets/Online+Retail} free of charge. We consider the use of such a publicly available dataset increases a reproducibility of a research process.
In machine learning field, researchers are encouraged to compare the performance of a newly proposed model or algorithm to old ones with a publicly available dataset, and the UCI repository is widely used in this aim. 


\subsection{Models}
We consider models of 2 by 2 design: static or dynamic demand, static or dynamic pricing. We denote each model as $m\in\{1,2,3,4\}$, where $m=1,2$ assume static demand, $m=3,4$ assume dynamic demand, $m=1,3$ assume static pricing, and $m=2,4$ assume dynamic pricing. For simplicity, we assume that the firm and consumers make their purchase decision independently across products. It is entirely possible to test if this assumption is valid or not using our CV algorithm, but we omit it as the main purpose of this section is an illustration of model selection procedure.
The consumers are heterogeneous in price sensitivity and the constant term of utility as in random coefficients model. We assume that consumers make a purchase at most once for each product within the considered period. This assumption is justified by the transaction level data. Among all the transactions used in the data, 75.8\% of them are made by consumers who purchased the same product only once in the considered period. An alternative approach is to model repeated purchase and inventory behavior explicitly as in \cite{hendel2006measuring}, but we do not take this path for tractability.

\subsubsection{Demand Model}
In each period, consumers in the market decide whether to purchase a product or not to maximize their objective function. If the demand is assumed to be static, the objective function is simply the utility function defined below. If the demand is dynamic, the objective function is the infinite-period sum of discounted utility. 

Denote products as $j=1,...,J$ and period as $t=1,...,T$. Consumer $i$'s utility of purchasing a product $j$ at period $t$ is
\begin{align*}
u_{ijt}&=\alpha^p_i p_{jt}+\alpha^0_{ij}+\mathbf{X}_{jt}\boldsymbol{\alpha}^x+\xi_{jt}+\epsilon_{ijt}.\\
& \equiv \delta_{ijt}+\epsilon_{ijt}
\end{align*}
where $p_{jt}$ is the price of a product $j$ in period $t$, $\mathbf{X}_{jt}$ is the observable characteristics, and $\xi_{jt}$ is the i.i.d preference shock, which enters the moment conditions. $\epsilon_{ijt}$ is the logit error term that  follows type-I extreme value distribution and i.i.d across periods and products. The utility of not purchasing is $u_{i0t}=\epsilon_{i0t}$ as the non-random component is normalized to be zero. The random coefficients follow a normal distribution.
\begin{align*}
\alpha^p_i = \alpha^p + \nu^p_{i} \rho^p\\
\alpha^0_i = \alpha^0 +\nu^0_{i} \rho^0
\end{align*}
,where $(\alpha^p,\alpha^0)$ are the population mean of the utility coefficients, $\nu^p_{i}$ and $\nu^0_{i}$ are draws from a standard normal distribution, and $(\rho^p,\rho^0)$ are the standard deviation of the distribution of the random coefficients.

In the static demand model, the consumers simply compare the utility of purchase to non-purchase in each period. Thus the purchase probability is 
\begin{align*}
s_{ijt}^{m} = \frac{\exp(\delta_{ijt})}{\exp(\delta_{ijt})+1}\\
\text{for } m=1,2.
\end{align*}

In the dynamic demand model, the consumers make purchase decision by comparing the instant utility to the value of waiting until next period. Let $\Omega_{ijt}^d$ be a state space for a consumer $i$ on product $j$ at period $t$ and $W_{ij}(\Omega_{ijt}^c)$ be a value function associated to the state. The Bellman equation is expressed as
\begin{align*}
W_{ij}(\Omega_{ijt}^d)=\max\{ u_{ijt},u_{i0t}+\beta\mathbb{E}[W_{ij}(\Omega_{ijt+1}^d) |\Omega_{ijt}^d] \}.  
\end{align*}

The purchase probability of product $j$ of a consumer $i$ at period $t$ is
\begin{align*}
s_{ijt}^{m}(\Omega^{d}_{ijt})=\frac{ \exp( \delta_{ijt} )}{ \exp( \delta_{ijt} )+ \exp( \beta \mathbb{E}[W(\Omega^{c}_{ijt+1})|\Omega^{c}_{ijt}] ) }\\
\text{for }m=3,4.
\end{align*}

Following \cite{conlon2012dynamic}, we make an assumption that consumers have perfect foresight over a transition of state $\Omega^{d}_{ijt}$. Formally, 
\begin{align*}
\mathbb{E}[W(\Omega^{d}_{ijt+1})|\Omega^{d}_{ijt}]=w_{ijt+1},
\end{align*}
where
\begin{align*}
w_{it} = \ln(\exp(\delta_{ijt})+\exp(\beta w_{it+1}))
\end{align*}
for all $i$, $j$, and $t$. The second line is a direct consequence of the first line following the argument of \cite{rust1987optimal}. An alternative and more popular specification is to assume that consumers form an expectation of the future state by certain functional form, typically an AR(1) regression. Compared to functional assumption perfect foresight reduces the computational burden significantly as it avoids integration over a distribution for calculating expectation (See \cite{conlon2012dynamic} for further discussion.) Also, note that by our CV algorithm we can even investigate which of perfect foresight and AR(1) assumption makes the model more accurate, which we believe is an interesting future work.

Finally, for both static and dynamic model let $M_{ijt}$ be the market size of consumers for a product $j$ at period $t$. Given the consumers purchase the same product at most once, the market size transition for any model $m\in\{1,2,3,4\}$ follows
\begin{align*}
M^{(m)}_{ijt+1} =M_{ijt}(1-s^{(m)}_{ijt}).
\end{align*}

\subsubsection{Supply Model}
We express the marginal cost of product $j$ at period $t$ for the retailer as $MC_{jt}$ where 
\begin{align*}
MC_{jt}=\mathbf{Y}_{j t}{\boldsymbol \gamma }_{jt} +\lambda_{jt}.
\end{align*}
$X^{cost}_{jt}$ is the observable characteristics of the product, and $\lambda_{jt}$ is the cost shock i.i.d across time and products. 

Denote the states of a product $j$ for the retailer at period $t$ as $\boldsymbol{\Omega}^s_{jt}$. $\boldsymbol{\Omega}^s_t$ includes the market size of each consumer segment $\{M_{ijt}\}_{i}$ and the draw of unobserved utility shock, $\{\xi_{ijt}\}_{i}$ and $\lambda_{jt}$. Given the demand system described above, the demand function is written as
\begin{align*}
D^{(m)}_{jt}(p_{jt},\boldsymbol{\Omega}^s_t) = \sum_{r=1}^R M^{(m)}_{ijt}s^{(m)}_{ijt}.
\end{align*}
The instant profit function of a product $j$ at period $t$ is therefore
\begin{align*}
\pi_{jt}(p_{jt},\boldsymbol{\Omega}^s_{jt})=D_{jt}(p_{jt},\boldsymbol{\Omega}^s_{jt}) (p_{jt} - MC_{jt}),
\end{align*}

In static pricing model, $m=1,3$, the retailer simply chooses the price to maximize the myopic profit:
\begin{align*}
p_{jt}^{m}=\argmax_{p_{jt}} \pi_{jt}(p_{jt}) \forall j,t\\
\text{for }m=1,3.
\end{align*}

In dynamic pricing model ($m=2,4$), the retailer maximizes the net profit over time with discounting. The discounting factor $\beta$ is assumed to be same with consumers. The retailer determines the price after observing the realization of the shocks, $\{\xi_{ijt}\}_{i}$ and $\lambda_{jt}$. The value function of a product $j$ is expressed as
\begin{align*}
V_j(\Omega^s_{t})=\mathbb{E}\left[\max_{p_{jt}} 
 \left(  \pi_{jt} + \beta  V_j(\Omega^s_{t+1}) \right)
 \middle|\mathbf{\Omega}^f_{t}, p_{jt} \right],
\end{align*}
where the expectation is over the unobserved cost shock in the next period, $\lambda_{jt+1}$. The optimal price is determined as
\begin{align*}
p_{jt}^{m}=\argmax_{p_{jt}} \pi_{jt}(p_{jt})+\beta \Expect[V_j(\Omega^s_{t+1})|\Omega^s_{t}]\\
\text{for }m=2,4. 
\end{align*}

Similar to the demand side, we assume that the retailer has a perfect information on the transition of the error draw. 


\subsubsection{Equilibrium}

This section describe the equilibrium condition for each model. When consumers and firms are both static ($m=1$), the equilibrium price and demand are the standard one as in many models such as \cite{berry1995automobile}. When consumers are static but firms are dynamic ($m=2$), pricing can be seen as a single agent dynamic optimization problem with continuous choice variable $p_{jt}$. Similarly, when consumers are dynamic but firms are static ($m=3$), consumers solve a single agent dynamic optimization problem. The consumers problem is an optimal stopping problem as the choice is the timing of purchase. When both consumers and the retailer are both dynamic ($m=4$), we assume their behavior is at Markov Perfect Nash Equilibrium (MPNE) where consumers' and retailer's prediction of the value function matches to the realization.

\subsection{Data}
We obtain our data from UCI machine Learning Repository. The UCI Machine Learning Repository maintains more than three hundreds datasets that are intensely used by machine learning community for empirical investigation and comparison of algorithms. When researchers propose a new model or algorithm in machine learning field, a common practice is to test its performance on the dataset in this repository. Such a culture gives a thorough idea on the practical performance of existing models and algorithms. Moreover, it helps a new researcher replicate the results on the existing papers. 


The dataset we utilize in this study is the online retail data created by \cite{chen2012data}, posted on UCI Machine Learning Repository in November 2015. The data is publicly available at \url{https://archive.ics.uci.edu/ml/datasets/Online+Retail}. The information about the data source is provided by the authors as follows: "The online retailer under consideration is a UK-based and registered non-store business with some 80 members of staff. The company was established in 1981 mainly selling unique all-occasion gifts. For years in the past, the merchant relied heavily on direct mailing catalogs, and orders taken over phone calls. It was only 2 years ago that the company launched its own web site and shifted completely to the web. Since then the company has maintained a steady and healthy number of customers. The company also uses Amazon.co.uk to market and sell its products."

The data include all the transactions occurred on this retailer from December 2010 to December 2011. Each transaction information includes quantity, unit price, consumer ID, and country. We dropped any sales to outside UK. The majority of the sales is inside UK and non-UK sales has only limited amount (approximately 20\%.)  Since our purpose is to demonstrate application of CV model selection to static and dynamic models, we aggregate the data into a monthly sales of each product so that the data format follows typical market level data and we can apply commonly used economic models. The monthly sales is simply a sum of the quantity sold in a particular month. The monthly price is calculated as the average of the price of transaction occurred in each month weighted by the quantity. We omitted the products that have any zero sales in the considered months from the data.

On the top of price and sales data, the author hand-coded product category and subcategory based on the description of products. The categories include {\it Children}, {\it Decoration}, or {\it Kitchen}. The number of products as well as basic statistics are summarized in table \ref{tab:OnlinRetail-Summary}. Figure \ref{fig:OnlineRetail_PQ} shows the average of monthly price and quantity sold in each category. It shows that the dynamics is heterogeneous across categories. For instance, the price of products in {\it Gift} and {\it Decoration} show tendency to decline over periods, while {\it Home and Garden}  or {\it Candle} show more fluctuation.

\subsection{Estimation and Model Selection}
We implement model selection for the demand side and supply side sequentially. First we test if the demand is static or dynamic. Subsequently, we test if the pricing is static or dynamic, assuming the demand model chosen in the previous step. The endogenous variables such as the market size $M_{ijt}$ and the share $s_{ijt}$ are estimated in the demand side, and imported over to the supply side estimation. Importantly, we do not have to specify the pricing model on estimation of demand side by virtue of perfect foresight assumption. We treat the data in each category independently.

We adapt 3-fold cross validation, $(k,r)=(1,3)$. Because the data has a panel structure of products and periods, either the product-wise or period-wise split is possible. We adopt split based on products. That is, we split the products into three groups, and use two of them to estimate a model and use the last one for validation. 

\subsubsection{MPEC formulation}
To estimate each model by GMM-MPEC, we formulate the estimation as a minimization problem of GMM objective with equilibrium constraints based on the model described above. Under the assumptions we impose, the equilibrium constraints are convex and mostly either linear or quadratic. This fact ensures that we are able to find an optimal solution of the estimation problem.

First we describe the MPEC formulation of demand models. For the static demand model ($m=1,2$), the set of constraints are
\begin{align}
\begin{split}
&s_{ijt}^{m}=\frac{\exp(\delta_{ijt})}{\exp(\delta_{ijt})+1}\label{eq:StatDemandConst}\\
&D_{jt}^{m}=\sum_i M_{ijt}^{m}s_{ijt}^{m}\\
&\delta_{ijt}=\alpha^p_i p_{jt}+\alpha^0_{ij}+\mathbf{X}_{jt}\boldsymbol{\alpha}^x+\xi_{jt}\\
& \alpha^p_i = \alpha^p + \nu^p_{i} \rho^p\\
&\alpha^0_i = \alpha^0 +\nu^0_{i} \rho^0 \\
&M_{ijt}^{m}=M_{ijt-1}^{m}(1-s_{ijt}^{m}),
\end{split}
\end{align}
for all $(i,j,t)$.

For dynamic demand model, the constraints are similar except the consumers compare the purchase utility to the value of waiting until next period.
\begin{align}
\begin{split}
&s_{ijt}^{m}=\frac{\exp(\delta_{ijt})}{\exp(\delta_{ijt})+\exp(\beta w_{it+1} )}\label{eq:DynDemandConst}\\
&w_{it} = \ln(\exp(\delta_{it})+\exp(\beta w_{it+1}) )\\
&D_{jt}=\sum_i M_{ijt}^{m} s_{ijt}^{m}\\
&\delta_{ijt}=\alpha^p_i p_{jt}+\alpha^0_{ij}+\mathbf{X}_{jt}\boldsymbol{\alpha}^x+\xi_{jt}\\
& \alpha^p_i = \alpha^p + \nu^p_{i} \rho^p\\
&\alpha^0_i = \alpha^0 +\nu^0_{i} \rho^0 \\
&M_{ijt}^{m}=M_{ijt-1}^{m}(1-s_{ijt}^{m})
\end{split}
\end{align}
for all $(i,j,t)$. 

The model parameters to estimate are $\theta^d = (\alpha^p,\alpha^0,\rho^p,\rho^0)$. The data to input are the realized demand $D_{jt}$, the observed price $p_{jt}$, and the random draws $\nu^p_i$ and $\nu^0_i$. The predicted share $s_{ijt}$, the market size of each consumer type $M_{ijt}$, and the error draw $\xi_{jt}$ are the endogenous variables. In the dynamic demand model, the value function $w_{ijt}$ is also observation-specific endogenous variable to choose for the optimization.

We define the supply side estimation problem by the first order condition and the Bellman equation. By abusing notation, let $D_{jt}^m(p)$ as a demand function with respect to price in model $m$. The supply side equilibrium constraints of static pricing model is that the observed prices are chosen to maximize the instant profit:
\begin{align}
\begin{split}
&D_{jt}^m = \sum_i M_{ijt}^m s_{ijt}\label{eq:StatSupplyConst}\\
&MC_{jt}^m=X^s_{jt}\gamma+\lambda_{jt}\\
&p_{jt} = \argmax_{p} [D_{jt}^m(p)(p-MC_{jt}^m)]
\end{split}
\end{align}
for all $(i,j,t)$. 

Instead of the third line above, the dynamic pricing model includes Bellman equation:
\begin{align}
\begin{split}
&D_{jt}^m = \sum_i M_{ijt}^m s_{ijt}\label{eq:DynSupplyConst}\\
&MC_{jt}=X^s_{jt}\gamma+\lambda_{jt}\\
&p_{jt} = \argmax_{p} [D_{jt}^m(p)(p-MC_{jt})+\beta V_{jt+1}(\Omega^s_{jt+1})]\\
&V_{jt}(\Omega^s_{jt}) = \max_{p} [D_{jt}^m(p)(p-MC_{jt})+\beta V_{jt+1}(\Omega^s_{jt+1})].
\end{split}
\end{align}

The model parameters to estimate is $\theta^s=\gamma$. $MC_{jt}$, $\lambda_{jt}$, and the value function are observation-specific endogenous variables. $M_{ijt}$ and $s_{ijt}$ are estimated in the demand side as endogenous variables.

In both static and dynamic model, the constraint includes the retailer's optimization problem. We convert it to the first order condition when solving for the estimation. The  details are in the Appendix. 

The GMM objective is a function defined by moment conditions 
\begin{align*}
&\Expect [\xi Z]=0\\
&\Expect [\lambda Z]=0,
\end{align*}
where $Z$ is the instrumental variables. It includes category and subcategory dummies, period dummy, and the market size of consumer segments $\{M_{it}\}_{i}$. The market size information is correlated with price because it relates to the price elasticity. Since we assume that the unobserved shocks are not serially correlated, the market size at period $t$ is not correlated with the shocks in the same period. Further detail of the setting for estimation is described in the Appendix.

\subsection{Results}
Table \ref{tab:DP_result_MS} presents the cross validation score of each model. The second from the last column shows the demand model selected by CV. The last column exhibits the selected pricing model. One can see that the selected model varies across categories. On demand side, the data on {\it Children} {\it Decoration}, and {\it Kitchen} are explained better by the static model, while the dynamic model is preferred on other categories. On supply side, static pricing explained the data of {\it Crafts}, {\it Decoration}, and {\it Personal Item} better. 

The result of model selection is difficult to interpret. One could try to provide some intuition: For instance, the products that fits static demand model better may be the ones that consumers cannot make a consumption plan. On products where the retailer engages in static pricing, it may be due to certain circumstance that researchers do not observe, such as a contract with wholesaler or limitation of inventory. However, prior to observing the result of cross validation, it is hard to make an reliable and scientific argument and justification for any model to be realistic.  

The difficulty of interpretation in turn suggests that it is impractical for researchers to assume a certain model beforehand. Selecting a structural model based on intuition may severely bias the inference. To see the problem, \ref{tab:DP_result_PC} shows the estimated price coefficient in each category in different specification. While in some cases two models exhibit fairly similar result, in some cases such as {\it Candle} or {\it Party} the result is largely different. Therefore, we recommend that researchers cross validate their models whenever possible, unless they have a strong reason to believe in certain model.

\section{CONCLUSION}\label{7-conclusion}
In this paper, we have proposed a cross-validation approach to model selection when models are estimated via GMM criterion. Cross-validation procedure can be readily implemented in any existing economic models without much extra work for researchers. We have proved its asymptotic consistency, and Monte-Carlo experiments in both linear and non-linear model confirm that cross-validation outperforms in-sample comparison that economists traditionally practice.

We also proposed a way to apply cross-validation when models are estimated through MPEC. As its real-world application, we adapt our CV based model selection to test dynamic demand model and dynamic pricing model in an online-retailer data. We find a quite diverse result across product categories. Unexpectedly, even on the same retailer it is not consistent whether a dynamic model is preferred or not. As the implication of structural estimation largely depends on the assumed model, this result suggests that economists should cross-validate their structural models rather than appealing to  for reliability of their inference.

\clearpage
\bibliography{main.bib}

\begin{thebibliography}{}

\bibitem[\protect\citeauthoryear{Andrews}{Andrews}{1999}]{andrews1999}
Andrews, D. W.~K. (1999).
\newblock Consistent moment selection procedures for generalized method of
  moments estimation.
\newblock {\em Econometrica\/}~{\em 67\/}(3), 543--563.

\bibitem[\protect\citeauthoryear{Angrist and Pischke}{Angrist and
  Pischke}{2010}]{angrist2010}
Angrist, J.~D. and J.-S. Pischke (2010, June).
\newblock The credibility revolution in empirical economics: How better
  research design is taking the con out of econometrics.
\newblock {\em Journal of Economic Perspectives\/}~{\em 24\/}(2), 3--30.

\bibitem[\protect\citeauthoryear{Berry, Levinsohn, and Pakes}{Berry
  et~al.}{1995}]{berry1995automobile}
Berry, S., J.~Levinsohn, and A.~Pakes (1995).
\newblock Automobile prices in market equilibrium.
\newblock {\em Econometrica: Journal of the Econometric Society\/}, 841--890.

\bibitem[\protect\citeauthoryear{Berto Villas-Boas}{Berto
  Villas-Boas}{2007}]{berto2007vertical}
Berto Villas-Boas, S. (2007).
\newblock Vertical relationships between manufacturers and retailers: Inference
  with limited data.
\newblock {\em The Review of Economic Studies\/}~{\em 74\/}(2), 625--652.

\bibitem[\protect\citeauthoryear{Bonnet and Dubois}{Bonnet and
  Dubois}{2010}]{bonnet2010inference}
Bonnet, C. and P.~Dubois (2010).
\newblock Inference on vertical contracts between manufacturers and retailers
  allowing for nonlinear pricing and resale price maintenance.
\newblock {\em The RAND Journal of Economics\/}~{\em 41\/}(1), 139--164.

\bibitem[\protect\citeauthoryear{Bresnahan}{Bresnahan}{1987}]{bresnahan1987}
Bresnahan, T. (1987).
\newblock Competition and collusion in the american automobile industry: The
  1955 price war.
\newblock {\em Journal of Industrial Economics\/}~{\em 35\/}(4), 457--82.

\bibitem[\protect\citeauthoryear{Chen, Sain, and Guo}{Chen
  et~al.}{2012}]{chen2012data}
Chen, D., S.~L. Sain, and K.~Guo (2012).
\newblock Data mining for the online retail industry: A case study of rfm
  model-based customer segmentation using data mining.
\newblock {\em Journal of Database Marketing \& Customer Strategy
  Management\/}~{\em 19\/}(3), 197--208.

\bibitem[\protect\citeauthoryear{Coase}{Coase}{1972}]{coase1972durability}
Coase, R.~H. (1972).
\newblock Durability and monopoly.
\newblock {\em The Journal of Law and Economics\/}~{\em 15\/}(1), 143--149.

\bibitem[\protect\citeauthoryear{Conlon}{Conlon}{2012}]{conlon2012dynamic}
Conlon, C.~T. (2012).
\newblock A dynamic model of prices and margins in the lcd tv industry.
\newblock {\em unpublished\/}.

\bibitem[\protect\citeauthoryear{Dub{\'e}, Fox, and Su}{Dub{\'e}
  et~al.}{2012}]{dube2012improving}
Dub{\'e}, J.-P., J.~T. Fox, and C.-L. Su (2012).
\newblock Improving the numerical performance of static and dynamic aggregate
  discrete choice random coefficients demand estimation.
\newblock {\em Econometrica\/}~{\em 80\/}(5), 2231--2267.

\bibitem[\protect\citeauthoryear{Gautier and Kitamura}{Gautier and
  Kitamura}{2013}]{gautier2013nonparametric}
Gautier, E. and Y.~Kitamura (2013).
\newblock Nonparametric estimation in random coefficients binary choice models.
\newblock {\em Econometrica\/}~{\em 81\/}(2), 581--607.

\bibitem[\protect\citeauthoryear{Gowrisankaran and Rysman}{Gowrisankaran and
  Rysman}{2012}]{gowrisankaran2012dynamics}
Gowrisankaran, G. and M.~Rysman (2012).
\newblock Dynamics of consumer demand for new durable goods.
\newblock {\em Journal of political Economy\/}~{\em 120\/}(6), 1173--1219.

\bibitem[\protect\citeauthoryear{Hall}{Hall}{2005}]{hall2005}
Hall, A. (2005).
\newblock {\em Generalized Method of Moments}.
\newblock Advanced Texts in Econometrics. OUP Oxford.

\bibitem[\protect\citeauthoryear{Hall and Inoue}{Hall and
  Inoue}{2003}]{HallInoue2003}
Hall, A.~R. and A.~Inoue (2003, June).
\newblock {The large sample behaviour of the generalized method of moments
  estimator in misspecified models}.
\newblock {\em Journal of Econometrics\/}~{\em 114\/}(2), 361--394.

\bibitem[\protect\citeauthoryear{Hall and Pelletier}{Hall and
  Pelletier}{2011}]{hallpelletier2011}
Hall, A.~R. and D.~Pelletier (2011).
\newblock {Non-Nested Testing in Models Estimated via Generalized Method of
  Moments}.
\newblock Technical Report~2.

\bibitem[\protect\citeauthoryear{Hendel and Nevo}{Hendel and
  Nevo}{2006}]{hendel2006measuring}
Hendel, I. and A.~Nevo (2006).
\newblock Measuring the implications of sales and consumer inventory behavior.
\newblock {\em Econometrica\/}~{\em 74\/}(6), 1637--1673.

\bibitem[\protect\citeauthoryear{Hu, Xiao, and Zhou}{Hu
  et~al.}{2014}]{hu2014collusion}
Hu, W.-M., J.~Xiao, and X.~Zhou (2014).
\newblock Collusion or competition? interfirm relationships in the chinese auto
  industry.
\newblock {\em The Journal of Industrial Economics\/}~{\em 62\/}(1), 1--40.

\bibitem[\protect\citeauthoryear{Lee}{Lee}{2013}]{lee2013vertical}
Lee, R.~S. (2013).
\newblock Vertical integration and exclusivity in platform and two-sided
  markets.
\newblock {\em The American Economic Review\/}~{\em 103\/}(7), 2960--3000.

\bibitem[\protect\citeauthoryear{Luo}{Luo}{2015}]{luo2015operating}
Luo, R. (2015).
\newblock The operating system network effect and carriers dynamic pricing of
  smartphones.

\bibitem[\protect\citeauthoryear{Nair}{Nair}{2007}]{nair2007intertemporal}
Nair, H. (2007).
\newblock Intertemporal price discrimination with forward-looking consumers:
  Application to the us market for console video-games.
\newblock {\em Quantitative Marketing and Economics\/}~{\em 5\/}(3), 239--292.

\bibitem[\protect\citeauthoryear{Nevo}{Nevo}{2001}]{nevo2001measuring}
Nevo, A. (2001).
\newblock Measuring market power in the ready-to-eat cereal industry.
\newblock {\em Econometrica\/}~{\em 69\/}(2), 307--342.

\bibitem[\protect\citeauthoryear{Rivers and Vuong}{Rivers and
  Vuong}{2002}]{riversvuong2002}
Rivers, D. and Q.~Vuong (2002).
\newblock Model selection tests for nonlinear dynamic models.
\newblock {\em Econometrics Journal\/}~{\em 5\/}(1), 1--39.

\bibitem[\protect\citeauthoryear{Rust}{Rust}{1987}]{rust1987optimal}
Rust, J. (1987).
\newblock Optimal replacement of gmc bus engines: An empirical model of harold
  zurcher.
\newblock {\em Econometrica: Journal of the Econometric Society\/}, 999--1033.

\bibitem[\protect\citeauthoryear{Smith}{Smith}{1992}]{smith1992}
Smith, R.~J. (1992).
\newblock Non-nested tests for competing models estimated by generalized method
  of moments.
\newblock {\em Econometrica\/}~{\em 60\/}(4), 973--980.

\bibitem[\protect\citeauthoryear{Su and Judd}{Su and
  Judd}{2012}]{su2012constrained}
Su, C.-L. and K.~L. Judd (2012).
\newblock Constrained optimization approaches to estimation of structural
  models.
\newblock {\em Econometrica\/}~{\em 80\/}(5), 2213--2230.

\bibitem[\protect\citeauthoryear{Yang}{Yang}{2007}]{yang2007consistency}
Yang, Y. (2007).
\newblock Consistency of cross validation for comparing regression procedures.
\newblock {\em The Annals of Statistics\/}, 2450--2473.

\end{thebibliography}
\bibliographystyle{chicago}

\clearpage
\begin{algorithm}[H]
\caption{$(k,r)$-Cross Validation on GMM}
\label{alg_cv_gmm}
\begin{algorithmic}[1]
\STATE Input: Models $\{ \mathcal{M}_i\}$, data $\{v_t\}_{t=1,...,T}$.
\FOR {each model $\mathcal{M}_i$}
\FOR {each training data $\{v_t\}_{t\in N_S}$}
\STATE Estimate model parameters as $$\theta_S^{(i)}=\argmin_{\theta^{(i)}\in\Theta^{(i)}} Q_S^{(i)}(\theta^{(i)})$$
\STATE Calculate the score $Q_{S,\mathrm{valid}}^{(i)}(\theta_S^{(i)})$
\ENDFOR
\STATE Calculate the average score
\begin{equation*}
 Q_{\mathrm{valid}}^{(i)} = \frac{1}{{}_r C _k} \sum_{S \subset \{1,2,\dots,r\}: |S| = r-k} Q_{S, \mathrm{valid}}^{(i)}(\theta_S^{(i)})
\end{equation*}
\ENDFOR
\STATE Find the best model that exhibits the smallest $Q_{\mathrm{valid}}^{(i)}$ .
\end{algorithmic}
\end{algorithm}


\begin{figure}
    \centering
    \setlength{\subfigwidth}{.49\linewidth}
    \addtolength{\subfigwidth}{-.49\subfigcolsep}
    \begin{minipage}[t]{\subfigwidth}
    \begin{subfigure}{\subfigwidth}
        \includegraphics[width=\textwidth]{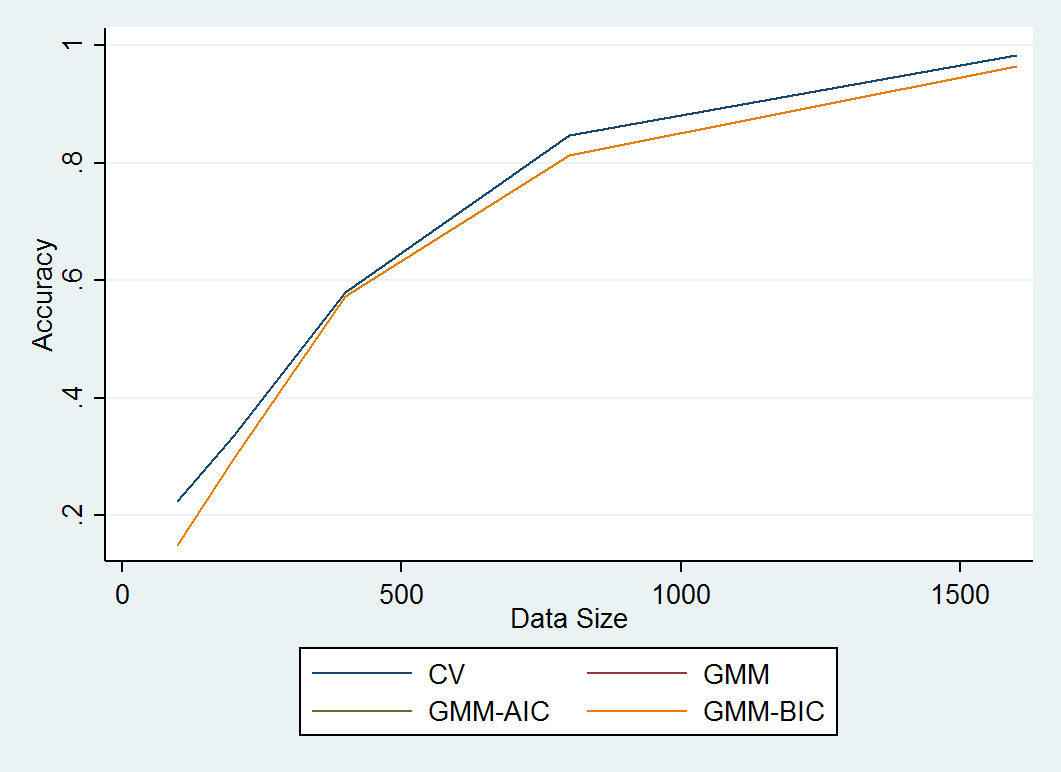}
        \caption{$p^1=3$,$p^2=5$,$\alpha=3.$}
        \label{fig: fig2a}
    \end{subfigure}
    \end{minipage}\hfill
    \begin{minipage}[t]{\subfigwidth}
    \begin{subfigure}{\subfigwidth}
        \includegraphics[width=\textwidth]{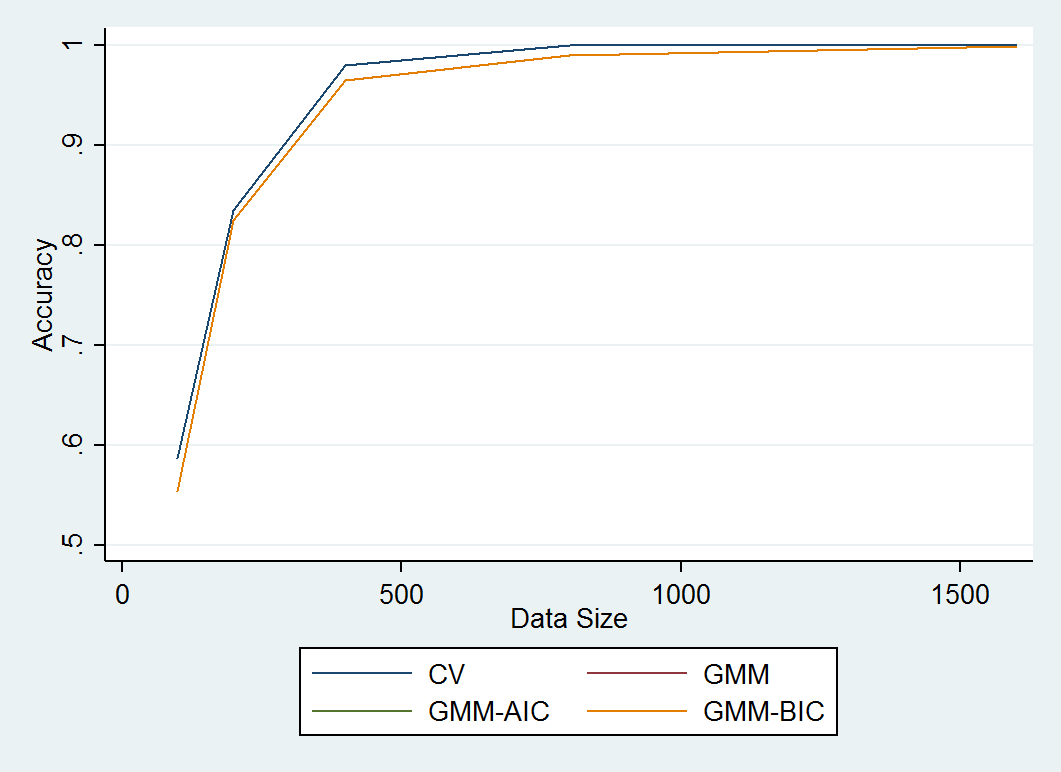}
        \caption{$p^1=3$,$p^2=5$,$\alpha=7.$}
        \label{fig: fig2b}
    \end{subfigure}       
    \end{minipage}\hfill
    \begin{minipage}[t]{\subfigwidth}
    \begin{subfigure}{\subfigwidth}
        \includegraphics[width=\textwidth]{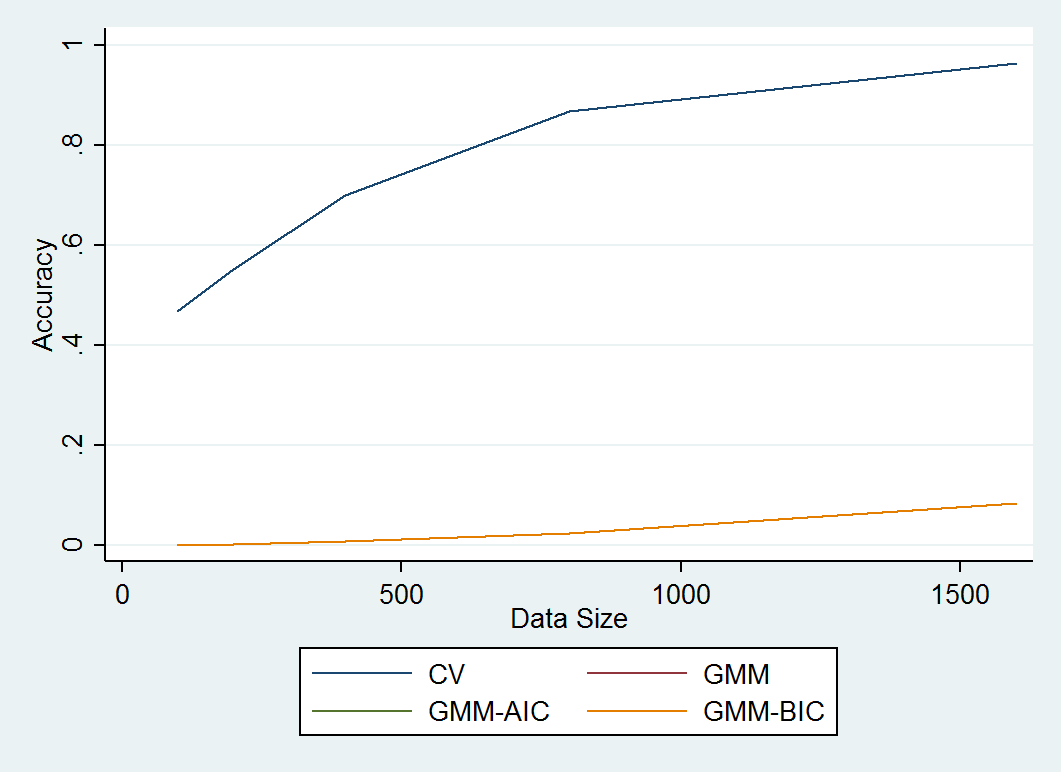}
        \caption{$p^1=3$,$p^2=9$, $\alpha=3.$}
        \label{fig: fig2c}
    \end{subfigure}
    \end{minipage}\hfill
    \begin{minipage}[t]{\subfigwidth}
    \begin{subfigure}{\subfigwidth}
        \includegraphics[width=\textwidth]{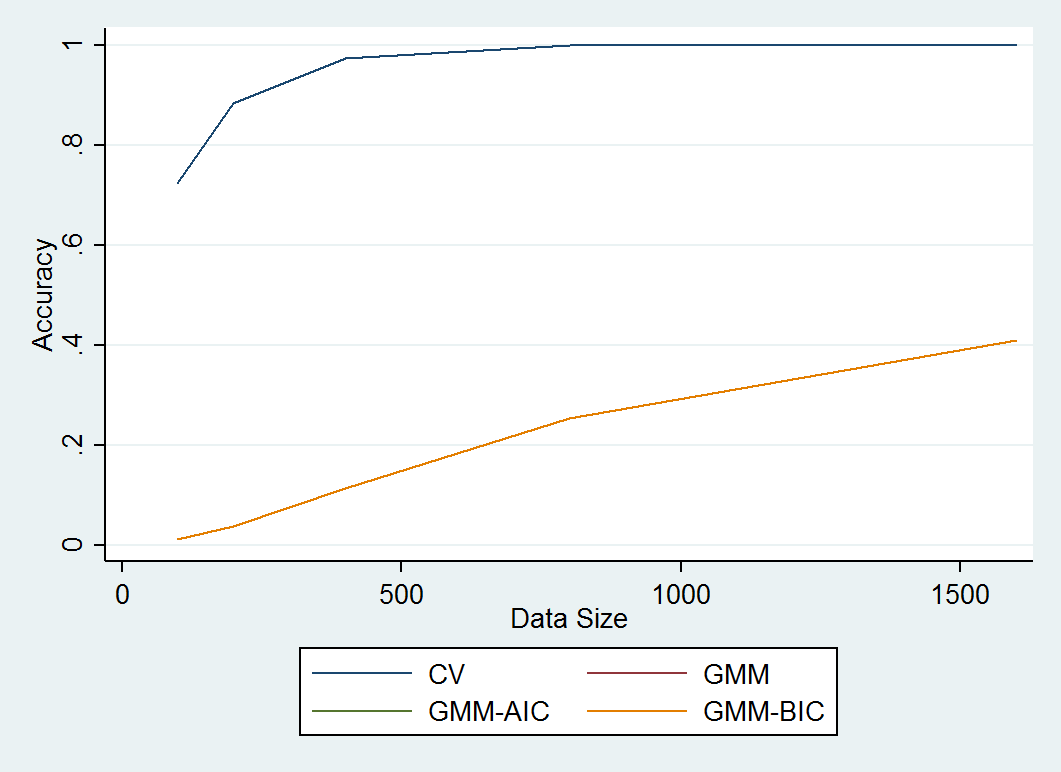}
        \caption{$p^1=3$,$p^2=9$, $\alpha=7.$}
        \label{fig: fig2d}
    \end{subfigure}
    \end{minipage}\hfill
    \caption{The accuracy of model selection when $p^1<p^2$. The $y$-axis is the probability that the correctly specified model (model 1) is chosen by each procedure. The number of instruments is set to be $c^1=c^2=10$. The cross-validation is 2-folds, i.e. $r=2$. The weighting matrix is set to be identity matrix.}\label{fig: overfit}
\end{figure}

\begin{figure}
    \centering
    \setlength{\subfigwidth}{.49\linewidth}
    \addtolength{\subfigwidth}{-.49\subfigcolsep}
    \begin{minipage}[t]{\subfigwidth}
    \begin{subfigure}{\subfigwidth}
        \includegraphics[width=\textwidth]{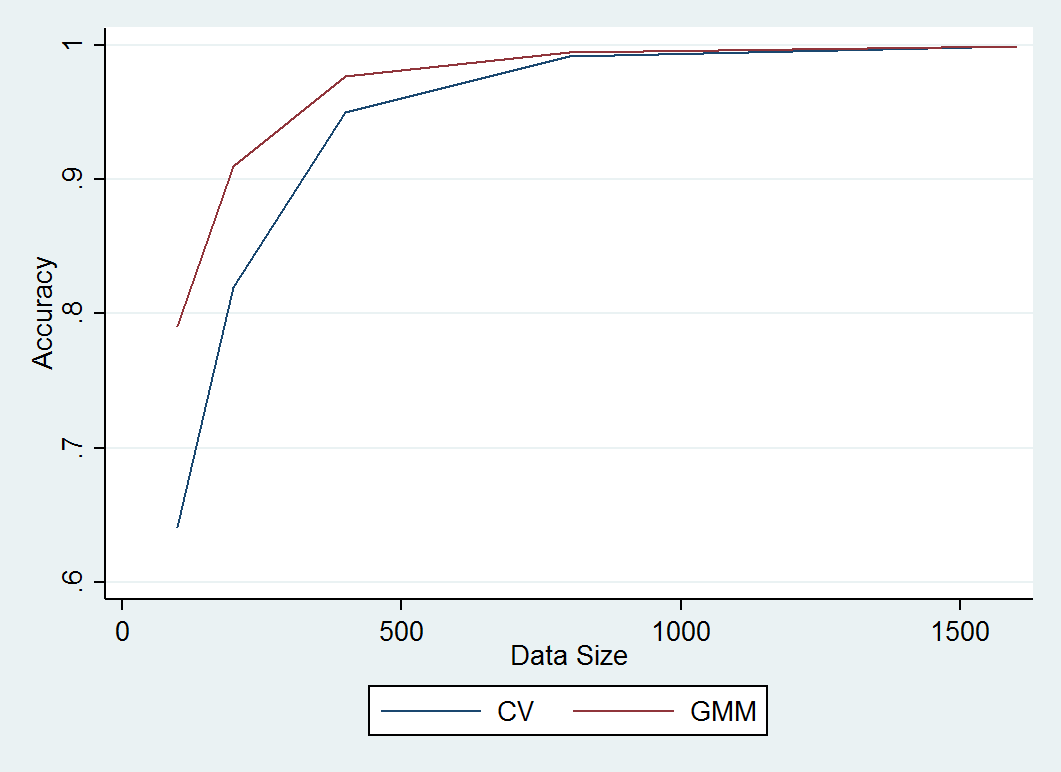}
        \caption{$p^1=p^2=3$,$\alpha=7.$}
        \label{fig: fig1a}
    \end{subfigure}
    \end{minipage}\hfill
    \begin{minipage}[t]{\subfigwidth}
    \begin{subfigure}{\subfigwidth}
        \includegraphics[width=\textwidth]{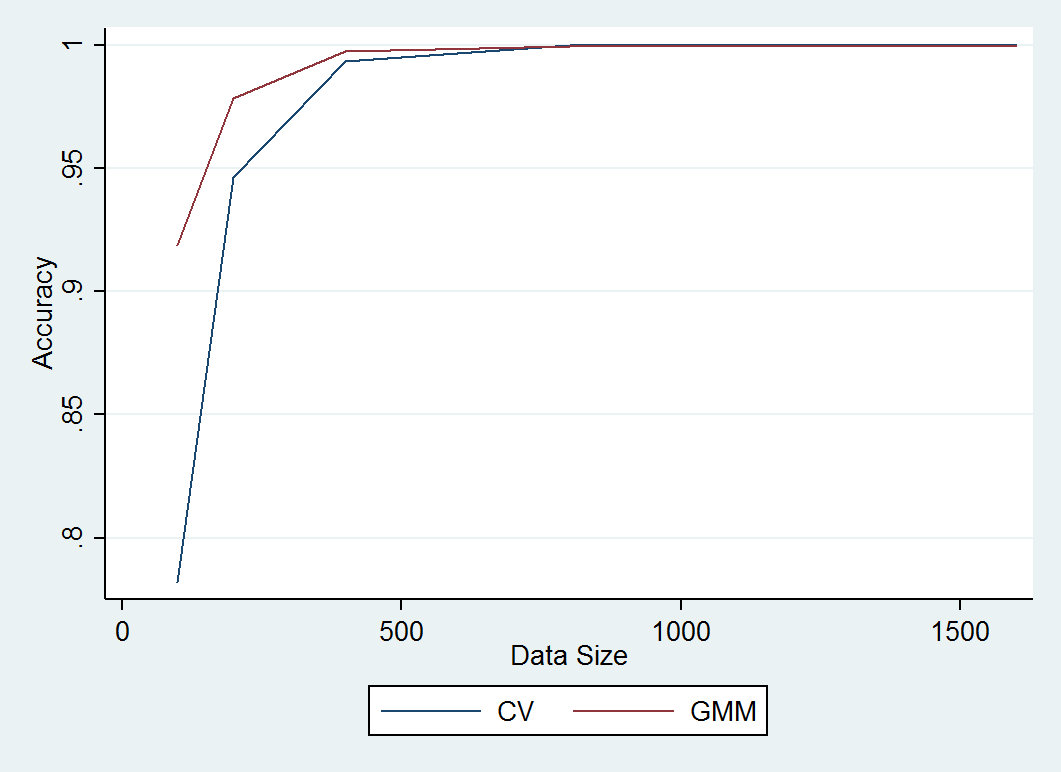}
        \caption{$p^1=p^2=3$,$\alpha=12.$}
        \label{fig: fig1b}
    \end{subfigure}       
    \end{minipage}\hfill
    \begin{minipage}[t]{\subfigwidth}
    \begin{subfigure}{\subfigwidth}
        \includegraphics[width=\textwidth]{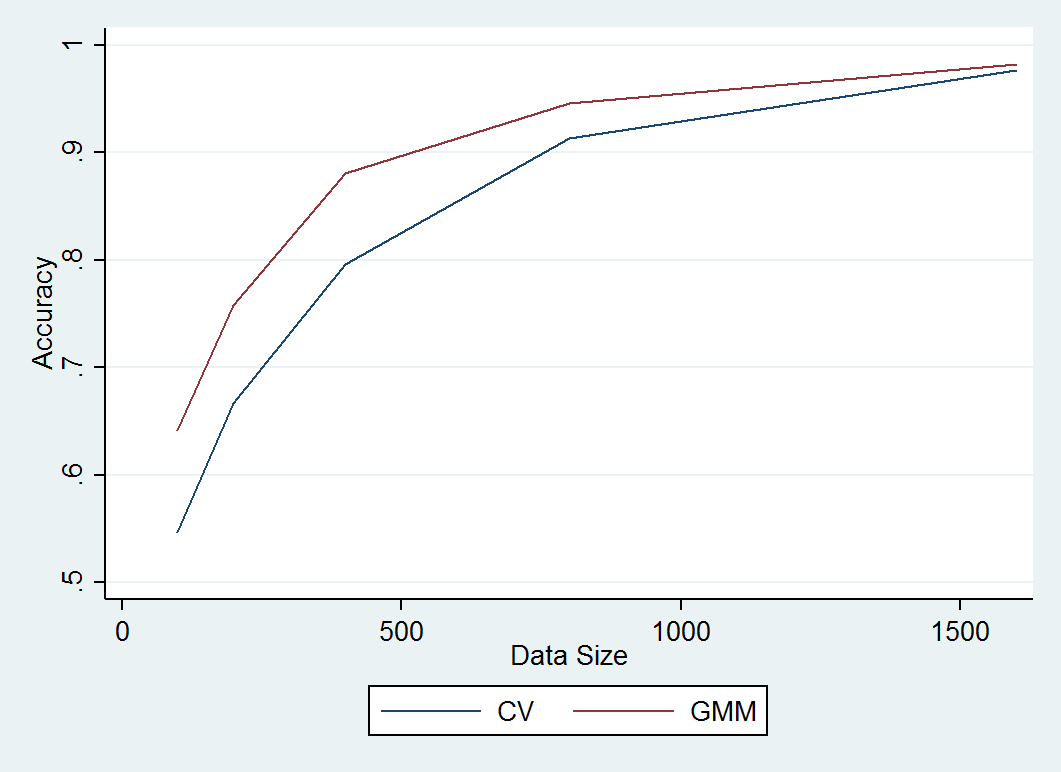}
        \caption{$p^1=p^2=7$, $\alpha=7.$}
        \label{fig: fig1c}
    \end{subfigure}
    \end{minipage}\hfill
    \begin{minipage}[t]{\subfigwidth}
    \begin{subfigure}{\subfigwidth}
        \includegraphics[width=\textwidth]{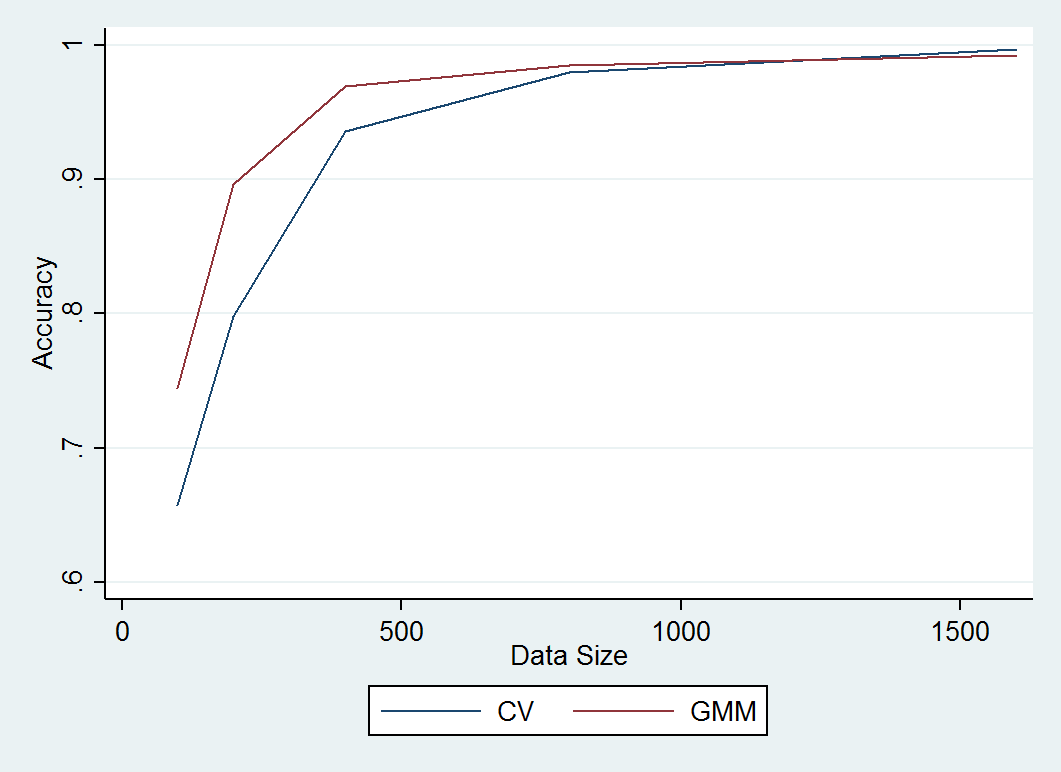}
        \caption{$p^1=p^2=7$, $\alpha=12.$}
        \label{fig: fig1d}
    \end{subfigure}
    \end{minipage}\hfill
    \caption{The accuracy of model selection when $p^1=p^2$. The $y$-axis is the probability that the correctly specified model (model 1) is chosen by each procedure. The number of instruments is set to be $c^1=c^2=10$. The cross-validation is 2-folds, i.e. $r=2$. The weighting matrix is set to be identity matrix.}\label{fig: p1eqp2_global}
\end{figure}

\begin{table}[htbp]
	\scriptsize
	\centering
	\caption{The validation Score of CV. Average of 100 iterations (standard deviation in the bracket). } \label{tab:CD_Score_CV}
\begin{tabular}{ccccccc}
 \hline\hline
& \multicolumn{6}{c}{$\alpha=-.1$} \\
 \cmidrule(r){4-5}
&  & \multicolumn{5}{c}{\textbf{Candidate Model}} \\
 \cmidrule(r){3-7}
Number of Market & \textbf{True Model} & $\{1,2,3\}$  & $\{1,2\}\{3\}$ & $\{1\}\{2,3\}$ & $\{1,3\}\{2\}$ & $\{1\}\{2\}\{3\}$ \\
\hline
\multirow{6}{*}{$25$} & \multirow{2}{*}{$\{1,2,3\}$} & $\mathbf{1.175}$  & $23.732$ & $28.822$ & $28.799$ & $30.709$ \\
 &  & (1.570)  & (44.224) & (89.920) & (56.014) & (60.764) \\
 & \multirow{2}{*}{$\{1,2\}\{3\}$} & $41.026$  & $\mathbf{1.022}$ & $27.687$ & $22.208$ & $8.799$ \\
 &  & (61.458)  & (1.049) & (54.927) & (29.543) & (13.536) \\
 & \multirow{2}{*}{$\{1\}\{2\}\{3\}$} & $25.441$  & $10.115$ & $8.536$ & $8.659$ & $\mathbf{0.912}$ \\
 &  & (26.184)  & (17.287) & (9.610) & (8.646) & (1.021) \\
\cmidrule(r){1-7}
\multirow{6}{*}{$50$} & \multirow{2}{*}{$\{1,2,3\}$} & $\mathbf{0.233}$  & $6.892$ & $6.505$ & $5.708$ & $6.686$ \\
 &  & (0.181)  & (10.774) & (5.495) & (5.073) & (7.486) \\
 & \multirow{2}{*}{$\{1,2\}\{3\}$} & $9.890$  & $\mathbf{0.314}$ & $5.328$ & $6.143$ & $2.466$ \\
 &  & (11.450)  & (0.207) & (3.391) & (5.696) & (1.781) \\
 & \multirow{2}{*}{$\{1\}\{2\}\{3\}$} & $10.050$  & $2.552$ & $3.764$ & $2.808$ & $\mathbf{0.274}$ \\
 &  & (10.486)  & (2.366) & (10.730) & (2.180) & (0.215) \\
\cmidrule(r){1-7}
\multirow{6}{*}{$75$} & \multirow{2}{*}{$\{1,2,3\}$} & $\mathbf{0.144}$  & $3.436$ & $3.470$ & $3.272$ & $3.580$ \\
 &  & (0.094)  & (2.328) & (2.232) & (1.865) & (2.415) \\
 & \multirow{2}{*}{$\{1,2\}\{3\}$} & $5.736$  & $\mathbf{0.170}$ & $3.284$ & $3.315$ & $1.367$ \\
 &  & (4.349)  & (0.114) & (2.456) & (1.864) & (0.867) \\
 & \multirow{2}{*}{$\{1\}\{2\}\{3\}$} & $6.651$  & $1.710$ & $1.800$ & $2.046$ & $\mathbf{0.190}$ \\
 &  & (5.723)  & (1.334) & (1.082) & (3.060) & (0.152) \\
\cmidrule(r){1-7}
\multirow{6}{*}{$100$} & \multirow{2}{*}{$\{1,2,3\}$} & $\mathbf{0.084}$  & $2.314$ & $2.475$ & $2.371$ & $2.374$ \\
 &  & (0.046)  & (1.195) & (1.716) & (1.418) & (1.424) \\
 & \multirow{2}{*}{$\{1,2\}\{3\}$} & $4.050$  & $\mathbf{0.124}$ & $2.289$ & $2.384$ & $0.951$ \\
 &  & (3.121)  & (0.071) & (1.754) & (1.785) & (0.671) \\
 & \multirow{2}{*}{$\{1\}\{2\}\{3\}$} & $4.463$  & $1.314$ & $1.266$ & $1.357$ & $\mathbf{0.124}$ \\
 &  & (2.282)  & (1.761) & (0.791) & (1.017) & (0.081) \\

\hline

& \multicolumn{6}{c}{$\alpha=-.3$} \\
 \cmidrule(r){4-5}
\hline
\multirow{6}{*}{$25$} & \multirow{2}{*}{$\{1,2,3\}$} & $\mathbf{1.139}$  & $2.531$ & $2.355$ & $1.906$ & $1.704$ \\
 &  & (2.269)  & (3.149) & (3.209) & (2.087) & (1.771) \\
 & \multirow{2}{*}{$\{1,2\}\{3\}$} & $8.223$  & $\mathbf{1.174}$ & $3.528$ & $4.746$ & $1.646$ \\
 &  & (13.529)  & (1.405) & (3.663) & (8.142) & (2.184) \\
 & \multirow{2}{*}{$\{1\}\{2\}\{3\}$} & $10.684$  & $3.272$ & $3.175$ & $4.273$ & $\mathbf{1.190}$ \\
 &  & (16.949)  & (4.735) & (3.757) & (9.404) & (1.373) \\
\cmidrule(r){1-7}
\multirow{6}{*}{$50$} & \multirow{2}{*}{$\{1,2,3\}$} & $\mathbf{0.281}$  & $0.651$ & $0.661$ & $0.640$ & $0.643$ \\
 &  & (0.256)  & (0.439) & (0.442) & (0.416) & (0.461) \\
 & \multirow{2}{*}{$\{1,2\}\{3\}$} & $1.713$  & $\mathbf{0.319}$ & $0.970$ & $1.165$ & $0.396$ \\
 &  & (1.448)  & (0.229) & (0.742) & (0.947) & (0.278) \\
 & \multirow{2}{*}{$\{1\}\{2\}\{3\}$} & $3.628$  & $0.998$ & $1.056$ & $1.041$ & $\mathbf{0.365}$ \\
 &  & (15.427)  & (1.742) & (2.198) & (1.949) & (0.607) \\
\cmidrule(r){1-7}
\multirow{6}{*}{$75$} & \multirow{2}{*}{$\{1,2,3\}$} & $\mathbf{0.159}$  & $0.387$ & $0.387$ & $0.426$ & $0.356$ \\
 &  & (0.110)  & (0.233) & (0.257) & (0.431) & (0.214) \\
 & \multirow{2}{*}{$\{1,2\}\{3\}$} & $1.096$  & $\mathbf{0.210}$ & $0.623$ & $0.574$ & $0.238$ \\
 &  & (0.684)  & (0.164) & (0.370) & (0.384) & (0.150) \\
 & \multirow{2}{*}{$\{1\}\{2\}\{3\}$} & $1.303$  & $0.504$ & $0.467$ & $0.464$ & $\mathbf{0.169}$ \\
 &  & (1.023)  & (0.336) & (0.288) & (0.322) & (0.117) \\
\cmidrule(r){1-7}
\multirow{6}{*}{$100$} & \multirow{2}{*}{$\{1,2,3\}$} & $\mathbf{0.103}$  & $0.261$ & $0.255$ & $0.277$ & $0.258$ \\
 &  & (0.060)  & (0.134) & (0.135) & (0.165) & (0.124) \\
 & \multirow{2}{*}{$\{1,2\}\{3\}$} & $0.743$  & $\mathbf{0.134}$ & $0.419$ & $0.414$ & $0.157$ \\
 &  & (0.393)  & (0.079) & (0.244) & (0.246) & (0.094) \\
 & \multirow{2}{*}{$\{1\}\{2\}\{3\}$} & $0.937$  & $0.317$ & $0.335$ & $0.329$ & $\mathbf{0.108}$ \\
 &  & (0.464)  & (0.170) & (0.243) & (0.253) & (0.068) \\
\bottomrule
\end{tabular}
\end{table}

\begin{table}[htbp]
	\scriptsize
	\centering
	\caption{The Model Selection Probability with CV. } \label{tab:CD_ChoiceProb_CV}
\begin{tabular}{ccccccc|c}
 \hline\hline
& \multicolumn{6}{c}{$\alpha=-.1$} \\
 \cmidrule(r){4-5}

&  & \multicolumn{5}{c}{\textbf{Candidate Model}} \\
 \cmidrule(r){3-7}
Number of Market & \textbf{True Model} & $\{1,2,3\}$  & $\{1,2\}\{3\}$ &  $\{1,3\}\{2\}$ & $\{1\}\{2,3\}$ & $\{1\}\{2\}\{3\}$ & \bf{true}  \\
\hline
\multirow{3}{*}{$25$} & $\{1,2,3\}$ & $\mathbf{0.99}$  & $0.00$ & $0.00$ & $0.01$ & $0.00$ & $0.99$\\
 & $\{1,2\}\{3\}$ & $0.00$  & $\mathbf{0.95}$ & $0.00$ & $0.00$ & $0.05$  & $0.95$\\
 &  $\{1\}\{2\}\{3\}$ & $0.00$  & $0.01$ & $0.00$ & $0.00$ & $\mathbf{0.99}$ & $0.99$\\
\cmidrule(r){1-7}
\multirow{3}{*}{$50$} & $\{1,2,3\}$ & $\mathbf{1.00}$  & $0.00$ & $0.00$ & $0.00$ & $0.00$ & $1.00$\\
 & $\{1,2\}\{3\}$ & $0.00$  & $\mathbf{0.99}$ & $0.00$ & $0.00$ & $0.01$ & $0.99$\\
 & $\{1\}\{2\}\{3\}$ & $0.00$  & $0.00$ & $0.00$ & $0.00$ & $\mathbf{1.00}$ & $1.00$\\
\cmidrule(r){1-7}
\multirow{3}{*}{$75$} & $\{1,2,3\}$ & $\mathbf{1.00}$  & $0.00$ & $0.00$ & $0.00$ & $0.00$ & $1.00$\\
 & $\{1,2\}\{3\}$ & $0.00$  & $\mathbf{1.00}$ & $0.00$ & $0.00$ & $0.00$ & $1.00$\\
 &  $\{1\}\{2\}\{3\}$ & $0.00$  & $0.00$ & $0.00$ & $0.00$ & $\mathbf{1.00}$ & $1.00$\\
\cmidrule(r){1-7}
\multirow{3}{*}{$100$} & $\{1,2,3\}$ & $\mathbf{1.00}$  & $0.00$ & $0.00$ & $0.00$ & $0.00$ & $1.00$\\
 & $\{1,2\}\{3\}$ & $0.00$  & $\mathbf{1.00}$ & $0.00$ & $0.00$ & $0.00$ & $1.00$\\
 &  $\{1\}\{2\}\{3\}$ & $0.00$  & $0.00$ & $0.00$ & $0.00$ & $\mathbf{1.00}$ & $1.00$\\

\hline

& \multicolumn{6}{c}{$\alpha=-.3$} \\
 \cmidrule(r){4-5}
 \hline
\multirow{3}{*}{$25$} & $\{1,2,3\}$ & $\mathbf{0.62}$  & $0.04$ & $0.04$ & $0.11$ & $0.19$ & $0.62$\\
 & $\{1,2\}\{3\}$ & $0.00$  & $\mathbf{0.62}$ & $0.01$ & $0.03$ & $0.34$  & $0.62$\\
 &  $\{1\}\{2\}\{3\}$ & $0.00$  & $0.07$ & $0.05$ & $0.09$ & $\mathbf{0.79}$ & $0.79$\\
\cmidrule(r){1-7}
\multirow{3}{*}{$50$} & $\{1,2,3\}$ & $\mathbf{0.77}$  & $0.05$ & $0.04$ & $0.10$ & $0.04$ & $0.77$\\
 & $\{1,2\}\{3\}$ & $0.00$  & $\mathbf{0.64 }$ & $0.01$ & $0.00$ & $0.35$ & $0.64$\\
 & $\{1\}\{2\}\{3\}$ & $0.00$  & $0.01$ & $0.04$ & $0.04$ & $\mathbf{0.91}$ & $0.91$\\
\cmidrule(r){1-7}
\multirow{3}{*}{$75$} & $\{1,2,3\}$ & $\mathbf{0.78}$  & $0.05$ & $0.07$ & $0.06$ & $0.04$ & $0.78$\\
 & $\{1,2\}\{3\}$ & $0.00$  & $\mathbf{0.57}$ & $0.02$ & $0.00$ & $0.41$ & $0.57$\\
 &  $\{1\}\{2\}\{3\}$ & $0.00$  & $0.05$ & $0.01$ & $0.03$ & $\mathbf{0.91}$ & $0.91$\\
\cmidrule(r){1-7}
\multirow{3}{*}{$100$} & $\{1,2,3\}$ & $\mathbf{0.82}$  & $0.04$ & $0.09$ & $0.02$ & $0.03$ & $0.82$\\
 & $\{1,2\}\{3\}$ & $0.00$  & $\mathbf{0.64}$ & $0.00$ & $0.00$ & $0.36$ & $0.64$\\
 &  $\{1\}\{2\}\{3\}$ & $0.00$  & $0.04$ & $0.02$ & $0.01$ & $\mathbf{0.93}$ & $0.93$\\
\bottomrule
\end{tabular}
\end{table}

\begin{table}[htbp]
	\scriptsize
	\centering
	\caption{The Model Selection Probability with GMM.  } \label{tab:CD_ChoiceProb_GMM}
\begin{tabular}{ccccccc|c}
 \hline\hline
& \multicolumn{6}{c}{$\alpha=-.1$} \\
 \cmidrule(r){4-5}

&  & \multicolumn{5}{c}{\textbf{Candidate Model}} \\
 \cmidrule(r){3-7}
Number of Market & \textbf{True Model} & $\{1,2,3\}$  & $\{1,2\}\{3\}$ &  $\{1,3\}\{2\}$ & $\{1\}\{2,3\}$ & $\{1\}\{2\}\{3\}$ & \bf{true}  \\

\multirow{3}{*}{$25$} & $\{1,2,3\}$ & $\mathbf{0.99}$  & $0.01$ & $0.00$ & $0.00$ & $0.00$ & $0.99$\\
 & $\{1,2\}\{3\}$ & $0.00$  & $\mathbf{0.95}$ & $0.00$ & $0.00$ & $0.05$  & $0.95$\\
 &  $\{1\}\{2\}\{3\}$ & $0.00$  & $0.02$ & $0.01$ & $0.02$ & $\mathbf{0.95}$ & $0.95$\\
\cmidrule(r){1-7}
\multirow{3}{*}{$50$} & $\{1,2,3\}$ & $\mathbf{1.00}$  & $0.00$ & $0.00$ & $0.00$ & $0.00$ & $1.00$\\
 & $\{1,2\}\{3\}$ & $0.00$  & $\mathbf{0.92}$ & $0.02$ & $0.00$ & $0.06$ & $0.92$\\
 & $\{1\}\{2\}\{3\}$ & $0.00$  & $0.00$ & $0.00$ & $0.03$ & $\mathbf{0.97}$ & $0.97$\\
\cmidrule(r){1-7}
\multirow{3}{*}{$75$} & $\{1,2,3\}$ & $\mathbf{1.00}$  & $0.00$ & $0.00$ & $0.00$ & $0.00$ & $1.00$\\
 & $\{1,2\}\{3\}$ & $0.00$  & $\mathbf{0.97}$ & $0.00$ & $0.00$ & $0.03$ & $0.97$\\
 &  $\{1\}\{2\}\{3\}$ & $0.00$  & $0.00$ & $0.00$ & $0.00$ & $\mathbf{1.00}$ & $1.00$\\
\cmidrule(r){1-7}
\multirow{3}{*}{$100$} & $\{1,2,3\}$ & $\mathbf{1.00}$  & $0.00$ & $0.00$ & $0.00$ & $0.00$ & $1.00$\\
 & $\{1,2\}\{3\}$ & $0.00$  & $\mathbf{0.96}$ & $0.00$ & $0.00$ & $0.04$ & $0.96$\\
 &  $\{1\}\{2\}\{3\}$ & $0.00$  & $0.01$ & $0.00$ & $0.00$ & $\mathbf{0.99}$ & $0.99$\\

\hline

& \multicolumn{6}{c}{$\alpha=-.3$} \\
 \cmidrule(r){4-5}
\hline
\multirow{3}{*}{$25$} & $\{1,2,3\}$ & $\mathbf{0.61}$  & $0.10$ & $0.09$ & $0.09$ & $0.11$ & $0.61$\\
 & $\{1,2\}\{3\}$ & $0.00$  & $\mathbf{0.66}$ & $0.00$ & $0.02$ & $0.32$  & $0.66$\\
 &  $\{1\}\{2\}\{3\}$ & $0.00$  & $0.06$ & $0.05$ & $0.05$ & $\mathbf{0.84}$ & $0.84$\\
\cmidrule(r){1-7}
\multirow{3}{*}{$50$} & $\{1,2,3\}$ & $\mathbf{0.69}$  & $0.10$ & $0.13$ & $0.07$ & $0.01$ & $0.69$\\
 & $\{1,2\}\{3\}$ & $0.01$  & $\mathbf{0.54}$ & $0.01$ & $0.03$ & $0.41$ & $0.54$\\
 & $\{1\}\{2\}\{3\}$ & $0.00$  & $0.02$ & $0.04$ & $0.05$ & $\mathbf{0.89}$ & $0.89$\\
\cmidrule(r){1-7}
\multirow{3}{*}{$75$} & $\{1,2,3\}$ & $\mathbf{0.77}$  & $0.09$ & $0.07$ & $0.06$ & $0.01$ & $0.77$\\
 & $\{1,2\}\{3\}$ & $0.00$  & $\mathbf{0.44}$ & $0.00$ & $0.00$ & $0.56$ & $0.44$\\
 &  $\{1\}\{2\}\{3\}$ & $0.00$  & $0.03$ & $0.00$ & $0.06$ & $\mathbf{0.91}$ & $0.91$\\
\cmidrule(r){1-7}
\multirow{3}{*}{$100$} & $\{1,2,3\}$ & $\mathbf{0.77}$  & $0.05$ & $0.11$ & $0.06$ & $0.01$ & $0.77$\\
 & $\{1,2\}\{3\}$ & $0.00$  & $\mathbf{0.39}$ & $0.01$ & $0.01$ & $0.59$ & $0.39$\\
 &  $\{1\}\{2\}\{3\}$ & $0.00$  & $0.01$ & $0.00$ & $0.02$ & $\mathbf{0.97}$ & $0.97$\\
\bottomrule
\end{tabular}
\end{table}

\begin{figure}
  \caption{The choice probability of true model on CV and GMM model selection.}\label{fig:CD_graph}
  \centering
    \includegraphics[width=1.\textwidth]{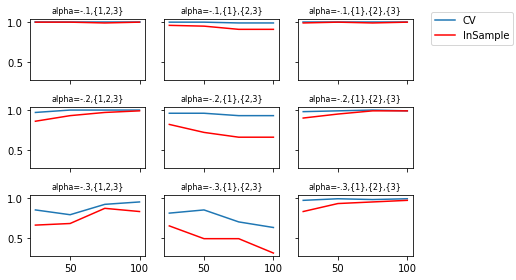}
\end{figure}

\begin{algorithm}[H]
\caption{$(k,r)$-Cross Validation on GMM-MPEC}
\label{alg_cv_gmm_mpec}
\begin{algorithmic}[1]
\STATE Input: Models $\{ \mathcal{M}_i\}$, data $\{v_t\}_{t=1,...,T}$.
\FOR {each model $\mathcal{M}_i$}
\FOR {each training data $\{v_t\}_{t\in N_S}$}
\STATE Estimate model parameters as
\begin{align*}
(\theta_S^{(i)},\sigma_S^{(i)},\eta_S^{(i)})&=\argmin_{\theta^{(i)},\sigma^{(i)},\eta^{(i)}} Q_S^{(i)}(\theta^{(i)},\sigma^{(i)},\eta^{(i)})\\
&\text{s.t. } h(\theta^{(i)},\sigma^{(i)},\eta^{(i)})=0.
\end{align*}
\STATE Calculate the score as 
\begin{align*}
Q_{S,\mathrm{valid}}^{(i)}(\theta_S^{(i)})=&\min_{\eta^{(i)}} Q_{\setminus S}(\theta_S^{(i)},\sigma_S^{(i)},\eta^{(i)})\\
&\text{s.t. } h(\theta_S^{(i)},\sigma_S^{(i)},\eta^{(i)})=0.
\end{align*}
\ENDFOR
\STATE Calculate the average score
\begin{equation*}
 Q_{\mathrm{valid}}^{(i)} = \frac{1}{{}_r C _k} \sum_{S \subset \{1,2,\dots,r\}: |S| = r-k} Q_{S, \mathrm{valid}}^{(i)}(\theta_S^{(i)})
\end{equation*}
\ENDFOR
\STATE Find the best model that exhibits the smallest $Q_{\mathrm{valid}}^{(i)}$ .
\end{algorithmic}
\end{algorithm}

\begin{table}[htbp]
	\scriptsize
	\centering
	\caption{Summary of Online-Retail Data} \label{tab:OnlinRetail-Summary}
\begin{tabular}{c|cccc}
 \hline\hline
 \multirow{2}{*}{Category} & \multirow{2}{*}{Example of Products} &\multirow{2}{*}{ \# Products} & Ave. Unit Price & Ave. Monthly Sales\\
  &  &  & (USD) & (Thousand)\\
\hline
Candle & Candles, Candle Holder, Candle Plate&77 & 1.944 & 0.232\\
Children & Baby Bib, Doll, Stationery Set & 175 & 4.122 & 0.148\\
Crafts & Knitting, Patches, Flannel, Sketchbook &38 & 2.694 & 0.214\\
Decoration & Photo frame, Flower, Decorative Signs&153 & 2.454 &  0.1954\\
Gift & Gift boxes, Tape, Message cards &65 & 0.7881 & 0.207\\
Home and Garden & Lamp,Cushion,Bath Salt &199 & 4.342 & 0.196\\
Kitchen & Mug, Tea Set, Lunch box&247 & 3.352 & 0.189\\
Party & Balloons,Napkins, Paper cup &75 & 2.432 & 0.197\\
Personal & Umbrella, Ring, Shopping bag & 109 & 2.864 & 0.159\\
\bottomrule
\end{tabular}
\end{table}

\begin{figure}
  \caption{The price and quantity dynamics of online retail data in each category.}\label{fig:OnlineRetail_PQ}
  \centering
    \includegraphics[width=1.\textwidth]{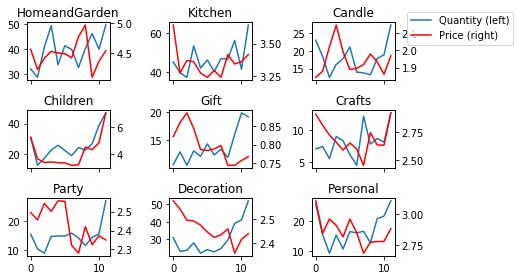}
\end{figure}

\begin{table}[htbp]
\caption{CV score in different categories}\label{tab:DP_result_MS}
\centering
\begin{tabular}{lcccc|cc}
\hline \hline
    &    \multicolumn{2}{c}{Demand}  & \multicolumn{2}{c}{Pricing} & \multicolumn{2}{c}{Selected Model} \\
            \cmidrule(r){2-3}\cmidrule(r){4-5}         \cmidrule(r){6-7}
Category   &  Static  &  Dynamic & Static & Dynamic  & Demand & Pricing \\
\hline
Candle &  .00683  &  .00651  & .015302  & .011402 & Dyn & Dyn \\ 
Children   &  .00913  &  17.9  & .982  & .376887 & Stat & Dyn \\ 
Crafts &  .00847  &  .00655 & .003628  & .004258 & Dyn & Stat \\ 
Decoration &  .00162  &  .00163  & .000454  & .000644 & Stat & Stat \\ 
Gift   &  .00328  &  .00277  & .000288  & .000119 & Dyn & Dyn \\
Home and Garden  &  .00177  &  .00109  & .053322  & 0.022654 & Dyn & Dyn \\ 
Kitchen &  .00152  &  .00158  & .002165  & .000763& Stat & Dyn \\ 
Party  &  .00795  &  .00310  & .016513  & .002486 & Dyn & Dyn \\ 
Personal Item &  .00305  &  .00193  & .003356  & 0.003443 & Dyn & Stat \\
\bottomrule
\end{tabular}
\end{table}

\begin{table}[htbp]
\caption{Estimated price coefficient in different categories}\label{tab:DP_result_PC}
\centering
\begin{tabular}{lcc}
\hline \hline
Category & Static model & Dynamic model\\
\hline
Candle & -6.52515  &  -1.71062 \\ 
Children   &  -0.01343  &  -0.01835 \\ 
Crafts &  -2.12241  &  -0.66246 \\ 
Decoration &  -1.26782  &  -0.83267  \\ 
Gift   &  -3.28775 &  -3.42115 \\
Home and Garden  &  -0.1841 &  -0.4384 \\ 
Kitchen &  -0.56447 &  -0.53658  \\ 
Party  &  -5.37787 &  -1.18728 \\ 
Personal Item &  -0.78579  &  -0.83831 \\
\bottomrule
\end{tabular}
\\ 
\end{table}


\clearpage
\begin{appendices}
\section{Detail of the simulation in section 4}
The dimension of the product characteristics on utility function $X_{jt}$ is set to be 2, where the first characteristic is constant and the second is randomly generated independently across products and periods. The cost side characteristics $Y_{jt}$ includes $X_{jt}$ and one additional characteristic also drawn independently. The characteristics are drawn from a normal distribution of mean $0$ and standard deviation $.1$. The unobserved error terms $\xi_{jt}$ and $\lambda_{jt}$ are also drawn from a normal distribution of mean zero and standard deviation $1$, independently across products and markets. The true values of parameters other than price coefficient $\alpha$ are $\beta=(2.,1.)$ and $\gamma=(3.,0.,1.)$. Those values are chosen to ensure that the marginal cost does not fall below zero, and the resulting share of outside option is not too close to zero for the invertibility of $\Delta$. Given the generated characteristics and the errors, the prices are simulated by solving the profit maximization problem by sequential least square quadratic programming. The results are robust to variety of parameter setting and distributional assumption.

\section{Detail of estimation procedure in section 6}
\subsection{Hyper parameter setting}
The discounting factor $\beta$ is set to be $.9$ for dynamic models both on demand and supply side. The draw of consumer types is generated from Halton sequence. The number of consumer segments is set to be 7. The initial market size $M_{ij1}$ is defined by the sum of the sales over the considered period in the subcategory that $j$ belongs to, divided by the number of consumer segments.

\subsection{Converting supply side constraints to FOC}
The equilibrium constraints on supply side includes the retailer's profit maximization.  In the estimation, we substitute it by first order condition. Let us define the derivative of the demand function with respect to price, $\pdif{D_{ijt}^m}{p_{jt}}$, to be another set of endogenous variable of MPEC that represents the derivative of the demand function from a consumer $i$ of a product $j$ at period $t$ evaluated at the realized price. Also define $\pdif{D_{jt}^m}{p_{jt}}$ be a derivative of the overall demand function, again at the observed price. In static pricing model, the MPEC constraints are converted to:
\begin{align*}
\begin{split}
&\pdif{D_{ijt}^m}{p_{jt}}=M_{ijt}^m s_{ijt}^m(1-s_{ijt}^m)\\
&\pdif{D_{jt}^m}{p_{jt}}=\sum_i \pdif{D_{ijt}(p_{jt})}{p_{jt}}\\
&D_{jt}(p_{jt})+\pdif{D_{jt}^m}{p_{jt}}
(p_{jt}-MC_{jt})=0\\
&MC_{jt}=X^s_{jt}\gamma+\lambda_{jt}\\
&\forall(j,t).
\end{split}
\end{align*}

In dynamic pricing model, FOC include a derivative of the value function. In addition to the ones above, we define two sets of additional endogenous variables: the realized value function of product $j$ at period $t$, $v_{jt}$, and the derivative of value function at next period with respect to current price evaluated at the observed price, $\pdif{V_{jt+1}}{p_{jt}}$. Then FOC and the Bellman equations translate to MPEC constraints:
\begin{align*}
\begin{split}
&D_{jt}(p_{jt})+\pdif{D_{jt}^m}{p_{jt}}
(p_{jt}-MC_{jt})+\beta \pdif{V_{jt}}{p_{jt}}=0\\
&v_{jt} = D_{jt}(p_{jt}-MC_{jt}) +\beta v_{jt+1}\\
&\forall(j,t).
\end{split}
\end{align*}
As we do not parametrically estimate the value function, the difficulty arises to calculate the derivative of the value function. The state variable at $t+1$ that are influenced by $p_{jt}$ are the market size of consumer segments $M_{ijt+1}$. Thus, define the derivative of the value with respect to market size, $\pdif{V_{jt+1}}{M_{ijt+1}}$, as another set of endogenous variable. Then,
\begin{align*}
\pdif{V_{jt+1}}{p_{jt}} &= \sum_i \pdif{V_{jt}}{M_{ijt+1}} \pdif{M_{ijt+1}}{p_jt}{p_{jt}}\\
&=\sum_i \pdif{V_{jt}}{M_{ijt+1}}\left( -\pdif{D_{ijt}(p_{jt})}{p_{jt}} \right).
\end{align*}
We still have to approximate $\pdif{v_{jt}}{M_{ijt+1}}$. One methodology is to utilize the estimated values of $v_{jt}$. The realized value $v_{jt}$ should be equal to the value function evaluated at the realized state $\Omega_{jt}$. Therefore, by comparing $v_{jt}$ and $M_{ijt}$, we are able to infer how value function changes with respect to $M_{ijt}$. In the estimation, we do so by linear approximation such as
\begin{align*}
\pdif{V_{jt}}{M_{ijt+1}} = \frac{1}{2}\left( \frac{v_{jt+1}-v_{jt}}{M_{ijt+1}-M_{ijt}}+\frac{v_{jt}-v_{jt-1}}{M_{ijt}-M_{ijt-1}}   \right).
\end{align*}

\end{appendices}

\end{document}